\documentclass[10pt,twocolumn]{IEEEtran}
\usepackage{epsfig}
\usepackage{amsmath}
\usepackage{amsfonts}
\usepackage{amssymb}
\usepackage{color} 
\usepackage{float}
\usepackage{url}
\usepackage{verbatim}
\usepackage{setspace}
\usepackage{caption}
\usepackage{subcaption}
\usepackage{cite}
\usepackage{hyperref}

\newcommand{\C}{\mathbb{C}}
\newcommand{\Co}{\mathcal{C}}
\newcommand{\E}{\mathbb{E}}
\newcommand{\R}{\mathbb{R}}
\newcommand{\Set}{\mathcal{S}}
\newcommand{\N}{\mathcal{N}}
\newcommand{\RP}{\mathcal{R}}

\newcommand{\Log}{\textrm{log}}

\newcommand{\mb}{\mathbf}
\newcommand{\bs}{\boldsymbol}

\newcommand{\Prob}{\textrm{Pr}}
\newcommand{\inner}{\textrm{in}}
\newcommand{\out}{\textrm{out}}
\newcommand{\bad}{\textrm{bad}}

\newcommand{\ind}{\mathbf{1}}
\newtheorem{definition}{Definition}
\newtheorem{lemma}{Lemma}

\newtheorem{theorem}{Theorem}

\newcommand{\argmin}[1]{\underset{#1}{\operatorname{argmin}}} 
\newsavebox{\tempbox}

\title{Time-Varying Space-Only Codes for Coded MIMO}
\author{Dieter Duyck, Sheng Yang, Fambirai Takawira, Joseph J. Boutros, and Marc Moeneclaey
\thanks{D. Duyck and M. Moeneclaey are with Ghent University, St-Pietersnieuwstraat 41, 9000 Gent, Belgium, \{dduyck,mm\}@telin.ugent.be.}%
\thanks{Sheng Yang is  with  Sup\'elec, 3 rue Joliot-Curie, 91192 Gif sur Yvette, France, sheng.yang@supelec.fr}
\thanks{F. Takawira is with the University of Witwatersrand, Private Bag 3, Johannesburg, South Africa, fambirai.takawira@wits.ac.za}
\thanks{Joseph J.  Boutros is  with  Texas A\&M University  at Qatar,  PO  Box 23874  Doha, Qatar,  boutros@tamu.edu}
\thanks{D. Duyck  thanks the University of Kwazulu Natal where he is currently registered as an exchange student and where this work was performed.}
\thanks{This paper will be presented in part at ISCCSP 2012 \cite{duyck2012tvso} and ISIT 2012 \cite{duyck2012tvso2}.}
\thanks{\copyright 2011 IEEE. Personal use of this material is permitted. Permission from IEEE must be obtained for all other uses, in any current or future media, including reprinting/republishing this material for advertising or promotional purposes, creating new collective works, for resale or redistribution to servers or lists, or reuse of any copyrighted component of this work in other works.}
}

\begin{document}
\maketitle

\begin{abstract}
Multiple antenna (MIMO) devices are widely used to increase reliability and information bit rate. Optimal error rate performance (full diversity and large coding gain), for unknown channel state information at the transmitter and for maximal rate, can be achieved by approximately universal space-time codes, but comes at a price of large detection complexity, infeasible for most practical systems. We propose a new coded modulation paradigm: error-correction outer code with space-only but time-varying precoder (as inner code). We refer to the latter as Ergodic Mutual Information (EMI) code. The EMI code achieves the maximal multiplexing gain and full diversity is proved in terms of the outage probability. Contrary to most of the literature, our work is not based on the elegant but difficult classical algebraic MIMO theory. Instead, the relation between MIMO and parallel channels is exploited. The theoretical proof of full diversity is corroborated by means of numerical simulations for many MIMO scenarios, in terms of outage probability and word error rate of LDPC coded systems. The full-diversity and full-rate at low detection complexity comes at a price of a small coding gain loss for outer coding rates close to one, but this loss vanishes with decreasing coding rate.
\end{abstract}


\section{Introduction}

Multiple antennas (MIMO) has become an important means to combat channel fading as well as to increase channel capacity. Since the pioneering work \cite{fos1998olo, telatar1999com}, a great amount of research effort has been focused on characterizing the fundamental limits of MIMO communications that are now well understood in many aspects. In particular, when the channel is subject to slow fading, a fundamental tradeoff as to how to optimally exploit the multiple antennas has been characterized in \cite{tse2003dam}. 

Code design problems for multi-antenna channels without channel state information at the transmitter have first been addressed by Tarokh et al. in \cite{tarokh1998stc}, in which the notion of space-time codes was introduced. By spanning the message over both spatial and temporal dimensions, full spatial diversity can be achieved. The elegant construction of orthogonal space-time codes enables a simple decoding at the receiver side \cite{tarokh1998stc, Alamouti1998ast}. Unfortunately, the benefit of simplicity is obtained at the cost of rate; more specifically, at most one independent symbol per channel use can be transmitted. Furthermore, the rate of orthogonal space-time codes decreases with the number of transmit antennas, while the ergodic capacity of a MIMO channel linearly increases with the minimum of the number of receive and transmit antennas. While the highly structured codes 
have a poor spectral efficiency, the random coding argument \cite{tse2003dam, elgamal2004lca} suggests that there exists a space-time code with finite length that can achieve the fundamental diversity-multiplexing tradeoff~(DMT). Guided by the diversity-multiplexing tradeoff, approximately universal \cite{tavildar2006auc}, i.e., DMT optimal for all channel statistics, structured space-time codes designs have been proposed \cite{belfiore2005tgc, oggier2006pstb, elia2006emd, tavildar2006auc}. The key to obtain approximately universal space-time codes with a finite code length is the non-vanishing properties of the product of the $\min\{n_T, {n_R}\}$ smallest singular values of the codeword matrices for increasing signal-to-noise ratio, where $n_T$ and ${n_R}$ are the number of transmit and receive antennas, respectively. In particular, this criterion coincides with the non-vanishing determinant criterion when $n_T={n_R}$ \cite{oggier2006pstb, elia2006emd}. Interestingly enough, when the MIMO channel is diagonal, i.e., a parallel channel, the criterion is reduced to the product-distance criterion, well known for code design for single-antenna block fading channels \cite{boutros1998ssd, bayer2004nac, tse2005fwc}. 

In general, the MIMO channel can be decomposed by the singular value decomposition into a diagonal channel preceded by a random unitary matrix rotating the channel input. Despite the diagonal channel being present in the singular value decomposition, there is a fundamental difference between MIMO and parallel channels when the channel state is unknown at the transmitter side. In fact, by carefully examining the criterion for approximate universality \cite{tavildar2006auc, elia2006emd}, constraints on the codeword structure are imposed in order to guarantee a ``good'' performance for the worst-case rotation incurred by the channel matrix. Such a phenomenon does not exist for diagonal/parallel channels. To guarantee good performance in any MIMO channel, precoding over at least $n_T^2$ dimensions is needed for MIMO, while we only need to code over $n_T$ dimensions for parallel channels. A direct consequence of the increased precoding dimensionality is the decoding complexity. While there exist efficient decoding algorithms for such problems~(see \cite{JaldenElia, natarajan2011gdl} and references therein), the complexity still grows in polynomial time with the dimension for the best-case scenarios and in exponential time for the worst cases.

The advantages of approximately universal space-time code designs are quite clear: high spectral efficiency, high diversity gains even without channel codes, and short block length. On the other hand, error-correction codes with relatively large block lengths are used in most modern communication systems. As a result, a common paradigm of coding for multi-antenna systems is the concatenation of an error correction code~(also referred to as outer code) and a space-time code~(also referred to as space-time precoder, or inner code). That is, a long coded block is split into many sub-blocks that are space-time precoded individually. While the precoder is designed without the consideration of an outer channel code, the actual error performance depends on multiple space-time blocks. Although using a well-designed space-time precoder can only improve the overall error rate performance, it comes at the price of high decoding complexity, as mentioned above. In fact, it turned out in some cases that the use of space-time codes in conjunction with channel codes does not bring extra gain but unnecessary decoding complexity (e.g. when the coding rate of the channel code is sufficiently low to recover a portion of the spatial diversity \cite{gresset2008stc}). Therefore, the conventional paradigm, to design space-time codes without considering the presence of an outer code, may not be efficient when it comes to the trade-off between performance and complexity. 

Inspired by the connection between MIMO and parallel channels (the MIMO channel is a parallel channel with random precoder) as well as the observation that the outer codeword spans over a large amount of channel uses, we propose a new coded modulation paradigm: error correction codes with space-only but time-varying precoder. The rationale behind this is to avoid the worst-case rotations (denoted as corrupt precoders in this paper) that are the main factor of deterioration in uncoded MIMO communication. The time-varying aspect of the space-only precoder reduces the effect of the worst-case rotations to a fraction of the transmission time, which can then be compensated by the outer code. Note that space-only precoders have been considered in the literature when channel state information was available at the transmitter (e.g. through feedback), see \cite{love2005lfu} and references therein.

The main contribution of this paper is as follows. We first observe that the MIMO channel corresponds to a parallel channel with a precoder that is random, but then remains fixed once it is chosen. We then extend the framework of signal space diversity for parallel channels to explain why full transmit diversity is not achieved with fixed space-only codes for MIMO. This new framework allows us to prove Theorem \ref{prop: full diversity mimo with distr rotation}, claiming that time-varying space-only precoders achieve full diversity on any MIMO channel, thereby proposing an alternative to approximately universal space-time codes with a much smaller detection complexity.

The rest of the paper is organized as follows. The system model, notation and problem statement are presented in Sec. \ref{System model} and we also clearly define what we mean by full-rate space-time codes. In Sec. \ref{sec: MIMO vs. parallel channels}, we interpret MIMO as parallel channels with a random precoder and explain that this randomness requires an extension of the study on signal space diversity for parallel channels. We first illustrate this extension by means of a toy example in Sec. \ref{sec: toy example}, where the channel gains, the precoder elements and the transmit symbols are all real-valued. Next, the extension of the signal space diversity framework is formalized in Sec. \ref{sec: Bad and Corrupt precoders for MIMO}, by introducing the concepts of bad and \textit{corrupt} precoders. This new framework is then used to prove that fixed space-only codes do not achieve full diversity (Sec. \ref{sec: FSOC}) in contrary to time-varying space-only precoders (Sec. \ref{sec: TVSOC}). Theorem \ref{prop: full diversity mimo with distr rotation} is then corroborated by presenting extensive numerical results in Sec. \ref{Numerical results}.


\section{System model}
\label{System model}

\textit{Notation:} we write scalars, vectors and matrices as $x$, $\mb{x}$ and $X$. $X^\dagger$ and $\mb{x}^\dagger$ are the Hermitian transposes of $X$ and $\mb{x}$. The Landau symbols $f(n) = O(g(n))$ and $f(n) = \Omega(g(n))$ respectively denote $f(n) \leq k g(n)$ and $f(n) \geq k g(n)$ for some positive $k$. The equation sign $f(\gamma) \doteq g(\gamma)$, introduced in \cite{tse2003dam}, is equivalent to $\lim_{\gamma \rightarrow \infty}\frac{\log f(\gamma)}{ \log \gamma} = \lim_{\gamma \rightarrow \infty} \frac{\log g(\gamma)}{ \log \gamma}$. Similar meanings hold for $\dot{\leq}$ and $\dot{\geq}$.

\subsection{Channel model}
\label{sec: Channel model}

We consider a point-to-point MIMO channel $H = [h_{i,j}] \in \C^{{n_R} \times n_T}$ with $n_T$ transmit antennas and ${n_R}$ receive antennas, where $h_{i,j} \sim \Co\N(0,1)$ is the complex path gain from transmit antenna $j$ to receive antenna $i$. We assume that all path gains are independent. The channel state information is perfectly known at the receiver side, but unknown at the transmitter side, i.e., no feedback channel is available. The channel is assumed to vary slowly, so that it remains constant during the transmission of at least one outer codeword. A channel use is referred to as the event where the transmitter sends $n_T$ symbols simultaneously from its $n_T$ transmit antennas. Assuming a total of $N_c$ channel uses per codeword, the discrete-time complex baseband equivalent channel equation is given by
\vspace{-0.1cm}
\begin{equation}
	\mb{\tilde{y}}_t = \sqrt{\gamma}\,H \mb{\tilde{x}}_t + \mb{\tilde{w}}_t,~~ t=1, \ldots, N_c,
\end{equation}
where $\mb{\tilde{y}}_t, \mb{\tilde{w}}_t \in \C^{{n_R} \times 1}$ denote the received vector and noise vector at instant $t$, and $\mb{\tilde{x}}_t \in \C^{n_T \times 1}$ denotes the symbol vector transmitted at instant $t$. The additive white Gaussian noise vector $\mb{\tilde{w}}_t$ has i.i.d. entries, $\tilde{w}_{t, i} \sim \Co\N(0,1)$. The transmit vector $\mb{\tilde{x}}_t$ satisfies $\E[||\mb{\tilde{x}}_t||^2] = n_T, \forall~ t$. This way, $\gamma$ is the average signal-to-noise ratio per symbol (SNR) per transmit antenna\footnote{In some papers, the average signal-to-noise ratio at each receive antenna is considered, which is $n_T \gamma$.}. The instantaneous mutual information is denoted as $I(\mb{\tilde{x}}_t; \mb{\tilde{y}}_t | H)$ and depends on the channel realization $H$.

The channel realization $H$ can be decomposed by a singular value decomposition as $H = U \Sigma V^\dagger$, where $U \in \C^{{n_R} \times {n_R}}$ and $V \in \C^{n_T \times n_T}$ are unitary matrices, uniformly distributed in the set of all unitary matrices (see App. \ref{app: Haar} for more background), and $\Sigma \in \R^{{n_R} \times n_T}$ is a diagonal matrix with the non-negative singular values $\sigma_i, i=1, \ldots, \min({n_R},n_T),$ of $H$ on its diagonal. As in \cite{tse2003dam, fab2007cmi}, we define the ordered normalized fading gains $\bs{\alpha} = [\alpha_1, \ldots, \alpha_{\min({n_R},n_T)}]$, where $\alpha_i = -\frac{\log \sigma_i^2}{\log \gamma}$, $\alpha_1 \geq \alpha_2 \geq \ldots \geq \alpha_{\min({n_R},n_T)}$. The joint distribution $p(\bs{\alpha})$ is given by \cite{tse2003dam, muirhead}
\begin{align}
 & p(\bs{\alpha}) = K^{-1} (\log{\gamma})^{n_T} \prod_{i=1}^{\min({n_R},n_T)} \gamma^{-(|{n_R}-n_T|+1)\alpha_i} \nonumber \\
 & \prod_{j>i} (\gamma^{-\alpha_i} - \gamma^{-\alpha_j})^2 e^{-\sum_{i=1}^{\min({n_R},n_T)} \gamma^{-\alpha_i}}.
 \label{eq: singular value distribution}
\end{align}

Because $H$ is known at the receiver, the following transformation can be performed,
\begin{equation}
	\mb{y}_t = U^\dagger \mb{\tilde{y}}_t = \sqrt{\gamma}~ \Sigma V^\dagger \mb{\tilde{x}}_t +
        \mb{w}_t,\quad t=1, \ldots, N_c
	\label{eq: new channel eq}
\end{equation}
where $\mb{w}_t$ follows the same distribution as $\mb{\tilde{w}}_t$. Because the transformation matrix $U^\dagger$ is invertible, no information is lost by this transformation (which can be proved by for example the data processing inequality \cite{cover2006eit})
, i.e., $I(\mb{\tilde{x}}_t; \mb{y}_t | \Sigma, V) = I(\mb{\tilde{x}}_t; \mb{\tilde{y}}_t | H)$.

Let us now consider (\ref{eq: new channel eq}) in the cases ${n_R}>n_T$ and ${n_R}<n_T$, where $\Sigma$ is not square. When ${n_R}>n_T$ and $\Sigma$ is tall, the last ${n_R}-n_T$ received symbols in $\mb{y}_t$ contain only noise due to the fact that the bottom ${n_R}-n_T$ rows of $\Sigma$ are zero. Hence, when we consider ${n_R}>n_T$ in the following, then, allowing an abuse in notation, $\mb{y}_t$, $\mb{w}_t$ and $\Sigma$ refer to the top $n_T$ rows of the actual vectors $\mb{y}_t$, $\mb{w}_t$ and $\Sigma$, respectively. As a result, an equivalent $n_T \times n_T$ MIMO channel is obtained, of course taking into account the actual singular value distribution as a function of ${n_R}$ and $n_T$ (see (\ref{eq: singular value distribution})). When ${n_R}<n_T$ and $\Sigma$ is fat, the last $n_T-{n_R}$ columns of $\Sigma$ are zero, so that the last $n_T-{n_R}$ columns of $V$ are not important, which will be considered in Sec. \ref{sec: TVSOC}.

The outage probability \cite{biglieri1998fci, ozarow1994itc} is expressed as
\begin{align}
	P_\out &= \Prob \left( \left[ \lim_{N_c \rightarrow \infty} \frac{1}{N_c} \sum_{t=1}^{N_c} I(\mb{\tilde{x}}_t; \mb{y}_t | \Sigma, V) \right] < R \right) \\
			&= \Prob \left( \E_t \left[ I(\mb{\tilde{x}}_t; \mb{y}_t | \Sigma, V) \right] < R \right),
	\label{eq: outage prob}
\end{align}
where $R$ is the spectral efficiency and $\E_t \left[.\right]$ denotes the temporal mean. The SNR-exponent of the outage probability, known as the diversity order, is 
\begin{equation}
	d_\out = \lim_{\gamma \rightarrow \infty} -\frac{\log P_\out}{\log \gamma},
	\label{eq: div order}
\end{equation}
where we assumed fading distributions where the limit in (\ref{eq: div order}) exists.

\subsection{Full rate space-time and space-only coding}
\label{Full rate space-time and space-only coding}

In this paper, full-rate STCs are considered where we define full-rate codes as schemes that transmit, on average, a linear transformation of $\min(n_T,{n_R})$ independent\footnote{Note that the dependence created by the outer error-correcting code is neglected when we use ``independent'' in this context.} $M$-QAM symbols at each channel use, corresponding to $m$ coded bits per channel use, where $m= \min(n_T,{n_R}) \log_2 M$. Hence, the number of channel uses per codeword is $N_c = \frac{N_b}{m}$, where $N_b$ is the length of a codeword of the outer code. 

Note that it is possible to convey $n_T$ independent symbols per channel use, which is larger than $\min(n_T,{n_R})$ when ${n_R}<n_T$. However, transmitting $\min(n_T,{n_R})$ independent symbols maximizes the multiplexing gain (the speed at which the constellation size may grow with the SNR $\gamma$) that can be achieved without leading to a degradation of the error rate for increasing SNR \cite{tse2003dam}. Secondly, the capacity of an $n_T \times {n_R}$ MIMO channel is the same as the capacity of an ${n_R} \times n_T$ MIMO channel (hence with ${n_R}$ transmit antennas and $n_T$ receive antennas) \cite{winters1987otc, fos1996lst, telatar1999com}, so that transmitting more symbols on the $n_T \times {n_R}$ MIMO channel than the maximum that can be conveyed on the ${n_R} \times n_T$ MIMO channel would not correspond to the limits of the channel, and thus lead to a significant loss in coding gain\footnote{Consider the simple case of a $1 \times 1$ channel where transmitting a linear combination of two independent symbols each channel use would also lead to a large loss in coding gain.}.

STCs are characterized by the number of space and time dimensions, $n_T$ and $k$, respectively, denoted as an $n_T \times k$ STC. When the number of time dimensions $k$ reduces to one, the STC becomes a space-only code. An $n_T \times k$ STC (consisting of $n_T k$ elements) yields $k$ symbol vectors $\{\mb{\tilde{x}}_{(l-1)k+1}, \ldots, \mb{\tilde{x}}_{(l-1)k+k}\}$, often represented through an $n_T \times k$ STC matrix $X_l$, where $X_{l,u,i}=\tilde{x}_{(l-1)k+u,i}, i=1, \ldots, n_T; u=1, \ldots, k; l=1, \ldots, N_c/k$. 

We first consider the case that ${n_R} \geq n_T$. The STC $X_l$ is obtained from $m k$ coded bits through a sequence of operations. The sequence of $m k$ coded bits is split in $n_T k$ groups of $m/n_T$ bits which are mapped to one of $M$ points, $M = 2^{m/n_T}$, belonging to an $M$-QAM constellation $\Omega_z$. Denoting as $\mb{z}_l$ the $n_T k$-dimensional symbol vector that results from mapping the $l$-th block of $m k$ coded bits, the linear precoding involves
\begin{equation}
\label{eq: precoding}
	\mb{\tilde{x}}_{\textrm{vec},l}  = P_l \mb{z}_l, ~ l=1, \ldots, N_c/k,
\end{equation}
where $P_l$ is a well chosen unitary matrix in $\C^{nk \times nk}$ and $\E[|z_{l,j}|^2] = 1, j=1, \ldots, n_T k$. Next, the $n_T k$-dimensional vector $\mb{\tilde{x}}_{\textrm{vec},l}$ is split in $k$ column vectors of length $n_T$, yielding the STC matrix $X_l$ or the $k$ vectors $\{\mb{\tilde{x}}_{(l-1)k+1}, \ldots, \mb{\tilde{x}}_{(l-1)k+k}\}$. Because $\E[|z_{l,j}|^2] = 1$ and $P_l$ is unitary, the components $\tilde{x}_{t,i}$ satisfy $\E[|\tilde{x}_{t, i}|^2] = 1, \forall~ i, t$. 

In the case that ${n_R} < n_T$, full-rate codes transmit, on average, a linear transformation of ${n_R}$ independent $M$-QAM symbols each channel use. Hence, $\mb{z}_l$ only contains ${n_R} k$ independent components. For example, $(n_T-{n_R})k$ components of $\mb{z}_l$ can be put to zero. To satisfy the constraint that $\E[||\mb{\tilde{x}}_t||^2] =n_T, \forall~t$, we set the mean square magnitude of the ${n_R} k$ non-zero components of $\mb{z}_l$ equal to $\frac{n_T}{{n_R}}$. 

\subsection{Problem formulation}

The approximately universal STCs that are optimal in terms of uncoded error rate for an $n_T \times {n_R}$ MIMO channel are $n_T \times n_T$ full-rate full-diversity STCs (hence $k=n_T$). The precoder $P_l = P$ is constant, has dimension $n_T^2 \times n_T^2$ and its $n_T^4$ elements are optimized to maximize the coding gain (see for example \cite{belfiore2005tgc,larsson2003stb,oestges2007mwc,tarokh1998stc,gresset2008stc,boutros2009tap}). A prohibitive objection is that the detection complexity of optimal STCs is very complex, e.g., it increases exponentially with $n_T^2$ for exhaustive ML-detection. For example, when $16-$QAM and $n_T=3$ transmit antennas are used, the detection complexity is $O(16^{9})$. Although a recent study shows that the complexity of near-ML decoding can be reduced to growing in polynomial time with the dimension for the best-case scenarios, by using lattice reduction and linear preprocessing \cite{JaldenElia}, it is yet to be verified whether the same conclusion holds with error correcting codes and soft decoding. 

We propose full-rate time-varying space-only codes, hence $k=1$ so that $P_l$ only has dimension $n_T \times n_T$, but $P_l$ varies with $l$. In Sec. \ref{sec: TVSOC}, we prove that full diversity is achieved. Because the STC is space-only, the notation simplifies: $\mb{\tilde{x}}_t = P_t \mb{z}_t$, $t=1, \ldots, N_c$, where $z_{t,i} \in \Omega_z$, $i=1, \ldots, \min(n_T,{n_R})$. Hence, $\mb{z}_t$ belongs to $\Omega_{\mb{z}} = (\Omega_z)^{n_T}$, which is the Cartesian product of $n_T$ constellations $\Omega_z$.
We denote our new precoder type by the EMI code, where EMI (\textit{Ergodic Mutual Information}) refers to the temporal mean of the mutual information. The detection complexity of exhaustive ML-detection now increases exponentially with $\min(n_T,{n_R})$ which is the lowest possible ML detection complexity for full-rate STCs. For example, when $\Omega_z =16-$QAM and $n_T={n_R}=3$, the detection complexity is $O(16^{3})$ for exhaustive ML-decoding, which is feasible.

The overall received signal-to-noise ratio per information bit is denoted as $\frac{E_b}{N_0}$, where $E_b$ is the ratio of the energy of the received symbol vector and the number of information bits conveyed per transmitted symbol vector, so that $\frac{E_s}{N_0} = \gamma = \frac{R E_b/N_0}{\E[||\mb{\tilde{x}}_t||^2] {n_R}}$, where $R = m R_c$ is the spectral efficiency and $R_c$ is the coding rate. When the constraint $\E[||\mb{\tilde{x}}_t||^2] =n_T, \forall~t$, is satisfied, then the SNR is $\gamma = \frac{R E_b/N_0}{n_T {n_R}}$.

\section{MIMO vs. parallel channels}
\label{sec: MIMO vs. parallel channels}

The channel equation (\ref{eq: new channel eq}) is that of a parallel channel $\Sigma$ (which corresponds to $\min(n_T,{n_R})$ parallel channels), with precoded input $V^\dagger \mb{\tilde{x}}_t$, with the particular feature that the precoder $V^\dagger$ is random. We denote the input of this parallel channel by 
\begin{equation}
\mb{x}_t = V^\dagger \mb{\tilde{x}}_t, 
\end{equation}
which belongs to a discrete constellation that is random (through $V$) and variable in time (through the time-varying precoder $P_t$). Allowing a small abuse in notation (by dropping the time-index), we denote this constellation by $\Omega_\mb{x}$, so that $\mb{x}_t \in \Omega_\mb{x}$. It is clear that $\mb{x}_t$ is a linear transformation of $\mb{z}_t$, 
\begin{equation}
	\mb{x}_t = V^\dagger P_t \mb{z}_t = V_t \mb{z}_t,
\end{equation}
where the precoder $V_t = V^\dagger P_t$ is random and varies in time within the duration of a codeword if $P_t$ is time-varying. When $P_t$ is deterministic or uniformly distributed in the set of all unitary matrices (denoted as $\mathcal{M}(n_T,n_T)$) and independent from $V$, then $V_t$ is also uniformly distributed in $\mathcal{M}(n_T,n_T)$ (see Lemma \ref{lemma: U properties} in App. \ref{app: Haar}). Hence, in the case that a fixed space-only precoder $P_t = P$ is used, the MIMO channel corresponds to a parallel channel with a random precoder $V^\dagger P$, and when our proposed \textit{time-varying} space-only precoder $P_t$ is used, the MIMO channel corresponds to a parallel channel with a random time-varying precoder $V_t$. 

For a given channel realization and a given channel use, $\Omega_\mb{x}$ is fixed and we can determine the mutual information between input and output of this parallel channel, $I(\mb{x}_t; \mb{y}_t | \Sigma, V) = I(\mb{\tilde{x}}_t; \bs{y}_t | H)$. Inserting $\mb{x}_t$ in Eq. (\ref{eq: new channel eq}), we have that 
\begin{equation}
	\mb{y}_t = \sqrt{\gamma}~ \Sigma \mb{x}_t + \mb{w}_t ,\quad t=1, \ldots, N_c,
	\label{eq: new channel eq 2}
\end{equation}
Note that only the top $\min(n_T,{n_R})$ components of the column vector $\mb{x}_t$ are important, by the structure of $\Sigma$; hence, the complete vector $\mb{x}_t$ is considered when ${n_R} \geq n_T$, and the top ${n_R}$ components of $\mb{x}_t$ otherwise. 

The mutual information $I(\mb{x}_t; \mb{y}_t | \Sigma, V)$ between input and output of a parallel channel is well known \cite{ungerboeck1982ccm, fab2007cmi} and recalled in Eq. (\ref{eq: mut info discrete alphabet MIMO form 1}), where $d^2(\mb{c}, \mb{d}) =\sum_{i=1}^{n_T} |c_i - d_i|^2 $. Using the normalized fading gains, the mutual information can be expressed as in Eq. (\ref{eq: mut info discrete alphabet MIMO}), where
\begin{equation}
f(\alpha_i, s_{t,i}, w_{t,i}) = e^{-\gamma^{1-\alpha_i} |s_{t,i}|^2 - 2 \sqrt{\gamma^{1-\alpha_i}} \RP\{w_{t,i} s_{t,i}^*\}},
\end{equation}
and where $s_{t,i} = (x_{t,i} - x_{t,i}^\prime)$ and $\RP\{.\}$ takes the real part. 
\begin{figure*}
\begin{align}
	I\left(\mb{x}_t; \mb{y}_t| \Sigma, V \right) &= m - 2^{-m} \sum_{\mb{x}_t \in \Omega_\mb{x}} \E_{\mb{y}_t|\mb{x}_t} \left[  \Log_2 \left( \sum_{\mb{x}_t^\prime \in \Omega_\mb{x}} \exp\left[ \left( d^2(\mb{y}_t,\sqrt{\gamma} \Sigma \mb{x}_t) - d^2(\mb{y}_t, \sqrt{\gamma}\Sigma \mb{x}_t^\prime) \right) \right] \right) \right] \label{eq: mut info discrete alphabet MIMO form 1} \\
	&=  m - 2^{-m} \sum_{\mb{x}_t \in \Omega_\mb{x}}  \E_{\mb{w}_t} \left[ \Log_2 \left( \sum_{\mb{x}_t^\prime \in \Omega_\mb{x}} \prod_{i=1}^{n_T} f(\alpha_i, s_{t,i}, w_{t,i}) \right) \right] \label{eq: mut info discrete alphabet MIMO}
\end{align}
\end{figure*}

The parallel channel model with precoding is well known for the study of signal space diversity (SSD) (see \cite{boutros1998ssd} for uncoded and \cite{fab2007mcm, duy2011pfo} for coded communication over parallel channels). In SSD, full transmit diversity is achieved when $V_t$ is chosen so that $s_{t,i} \neq 0, i=1, \ldots, \min(n_T,{n_R}), \forall~t$, where $\mb{s}_t = \mb{x}_t^\prime-\mb{x}_t; ~\mb{x}_t^\prime=V_t \mb{z}^\prime, \mb{x}_t=V_t \mb{z}; \mb{z}, \mb{z}^\prime  \in \Omega_\mb{z}$, $\mb{z}^\prime \neq \mb{z}$. Mostly, the fading gain distribution that is considered for parallel channels is Rayleigh fading, which yields a maximum transmit diversity of $\min(n_T,{n_R})$. However, the fading gain distribution of the singular values in $\Sigma$ is not Rayleigh \cite{muirhead, telatar1999com}, yielding a maximum transmit diversity of $nr$, which will be made more formal in Sec. \ref{sec: TVSOC}.

Consider a parallel channel with a constant deterministic precoder, $V_t=V_{\textrm{co}}$. For any coding rate $R_c$, it holds that full transmit diversity is not achieved when $V_{\textrm{co}}$ is a \textit{bad} precoder.
\begin{definition}
\label{def: bad precoders}
	We define \textit{bad} precoders $V_\bad$ as the set of precoders so that $\exists~i \in \{1,\ldots,\min(n_T,{n_R})\}, \mb{z}, \mb{z}^\prime \neq \mb{z}$, satisfying $s_{t,i}=0$.
\end{definition}
More importantly, if $V_{\textrm{co}}$ is not a \textit{bad} precoder, full transmit diversity is achieved (see \cite{boutros1998ssd} for more background). This is well known but it is particularly interesting for the following reason. 

Consider a constant space-only precoder $P_t = P$, so that $V_t = V^\dagger P = V^\prime$ is \textit{random} but not time-varying within the duration of a codeword. By Lemma \ref{lemma: U properties} in App. \ref{app: Haar}, the distributions of $V^\prime$ and $V$ are the same, hence the distributions of $H$ and $HP$ are the same. As a consequence, the space-only precoder $P$ achieves the same error rate performance as for uncoded MIMO (without precoding, or with $P=I$), which only achieves a diversity order of ${n_R}$ (which is proved in Sec. \ref{sec: FSOC}). Contrary to parallel channels, the loss of transmit diversity is not caused by bad precoders, which have a zero probability of occurrence due to the continuous space of the realizations of $V^\prime$.

Hence, the notion of bad precoders, established in the SSD framework for parallel channels with constant deterministic precoders, needs to be extended to the case of random precoders, to explain this loss of transmit diversity (also denoted as spatial diversity). This extension will be exemplified through a toy example (in which the constellation $\Omega_z$ and the precoder are taken real valued, allowing a geometrical illustration) in Sec. \ref{sec: toy example}, and will be formalized in Sec. \ref{sec: Bad and Corrupt precoders for MIMO} for MIMO. 

\section{Toy Example}
\label{sec: toy example}

For the toy example, we consider a classical system model with two flat non-ergodic parallel channels with Rayleigh fading and BPSK symbols at the precoder input (see for example \cite{boutros1998ssd, duy2011pfw}), which allows a geometric illustration. The system model is 
\begin{equation}
	\left\{
	\begin{array}{l}
		y_{t,1} = \sqrt{\gamma} \beta_1 x_{t,1} + w_{t,1} \\
		y_{t,2} = \sqrt{\gamma} \beta_2 x_{t,2} + w_{t,2} \\
	\end{array}
	\right., ~~t=1, \ldots, N_c,
	\label{eq: toy example channel eq}
\end{equation}
where $\beta_1$ and $\beta_2$ are i.i.d. and Rayleigh distributed, $\E[\beta_i^2]=1$, $[x_{t,1} ~ x_{t,2}]^T = Q [z_{t,1} ~ z_{t,2}]^T$ where $Q$ is a standard two-dimensional rotation matrix parametrized by the angle $\theta$ and $z_{t,i} \in \{ \pm 1\}$, and where $w_{t,i} \sim \N(0,0.5)$. Because each component of $\mb{x}_t$ is transmitted on another fading gain, this scheme is also denoted as component interleaving. The rotation of $\mb{z}$ is illustrated in Fig. \ref{fig: system model toy example}.
\begin{figure}
	\centering
	\includegraphics[width=0.4 \textwidth]{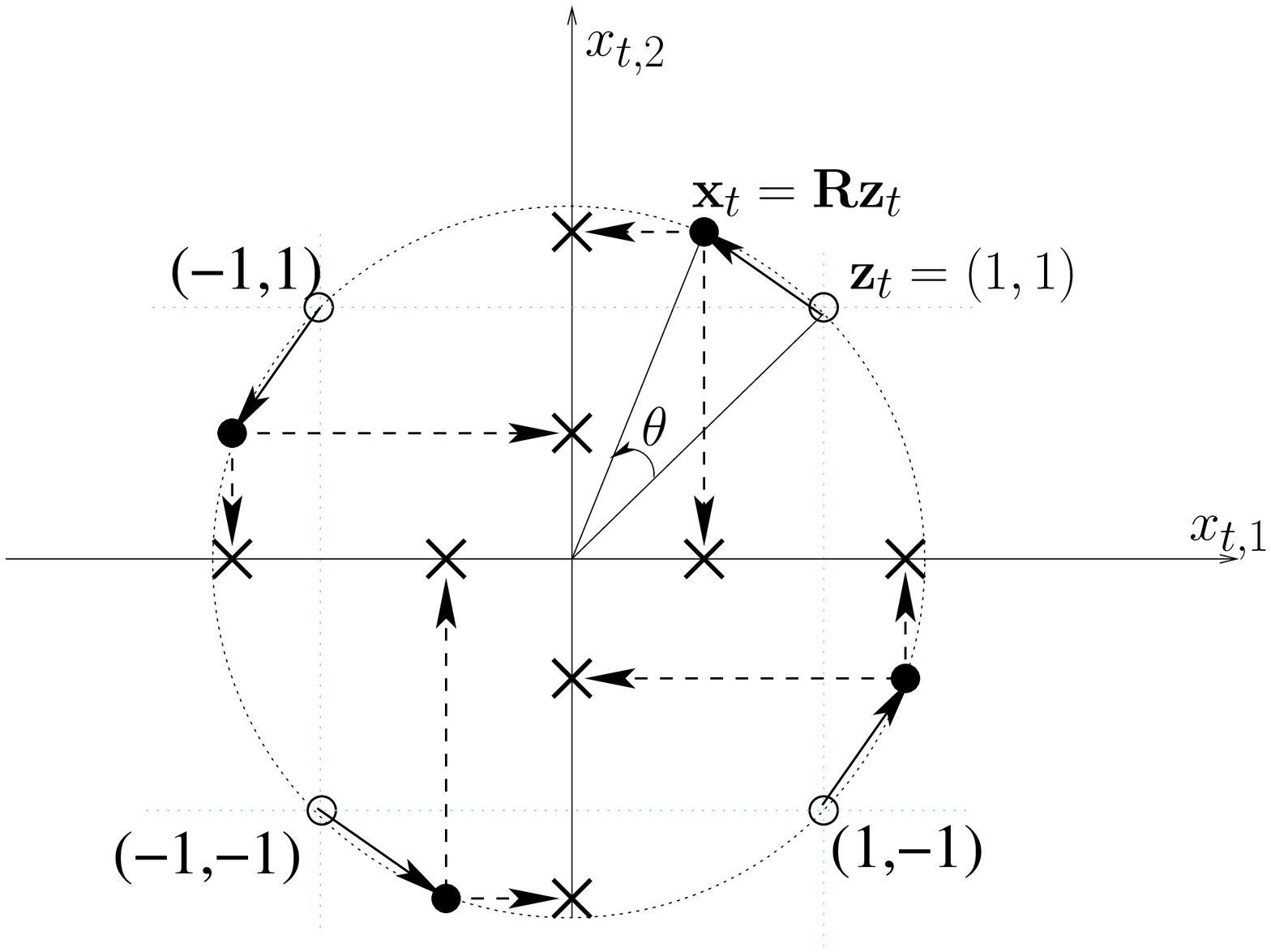}
	\caption{The system model of the toy example is illustrated. The coded bits are first mapped to a point in the constellation $\Omega_{\mb{z}}$ (empty circles), which is transformed to a point in the constellation $\Omega_{\mb{x}}$ (filled circles). The projections of the points of $\Omega_{\mb{x}}$ on the coordinate axes (marked by crosses) illustrates the component interleaving. The first and second component of $\mb{x}_t$, $x_{t,1}$ and $x_{t,2}$, $\forall t$, are affected by the fading gains $\beta_1$ and $\beta_2$, respectively.}
	\label{fig: system model toy example}
\end{figure}
According to Def. \ref{def: bad precoders}, bad precoders are rotation matrices $Q$ with rotation angle $\theta \in \theta_{\textrm{bad}} = \{k \pi/4, k=0, \ldots, 7\}$, in which case some crosses in Fig. \ref{fig: system model toy example} coincide. 

In the literature, only deterministic rotations were studied. Here, we study two other cases:
\begin{itemize}
	\item $\theta$ is random, but fixed once it is chosen (Sec. \ref{Random rotation});
	\item $\theta$ is random and independently generated at each channel use (Sec. \ref{Random time-varying rotation}).
\end{itemize}
The first case is similar to a non-ergodic MIMO channel, where the channel is random but remains constant during the transmission of an outer codeword. The second case is similar to a non-ergodic MIMO channel with a time-varying space-only precoder $P_t$ at its input.

\subsection{Random rotation}
\label{Random rotation}

It is well known \cite{boutros1998ssd, fab2007mcm, duy2011pfo} that a diversity order of two (denoted as full diversity in this section) is achieved for any coding rate $R_c \leq 1$ (including uncoded communication) when a fixed deterministic rotation matrix $Q$ is used and $s_{t,i}\neq 0, \forall~ t,i, \forall~ \mb{z}^\prime \neq \mb{z}$, where $\mb{s}_t = \mb{x}_t - \mb{x}_t^\prime$. In the context of MIMO, it is more interesting to consider a random rotation matrix $Q$. Despite the fact that $\Prob(\theta \in \theta_{\textrm{bad}})=0$ when $p(\theta)=\frac{1}{2 \pi}$, full diversity is not achieved as proved by the following lemma. 
\begin{lemma}
	In a point-to-point flat parallel Rayleigh fading channel as given in (\ref{eq: toy example channel eq}) with a fixed but random precoder $Q$ where $\theta$ follows a uniform distribution in $[0, 2 \pi]$, the diversity order can not be larger than $1.5$ for any coding rate $0.5 < R_c < 1$.
	\label{lemma: random prec toy ex}
\end{lemma}
\begin{proof}
	See App. \ref{sec: proof lemma toy example}.
\end{proof} 

From the proof of Lemma \ref{lemma: random prec toy ex}, it is clear that \textit{corrupt} precoders (and not \textit{bad} precoders) are the main cause of the full diversity loss. We formally define corrupt precoders in Sec. \ref{sec: Bad and Corrupt precoders for MIMO}, but for this section, consider corrupt precoders as the set of rotations where $\theta \in [0,\gamma^{-0.5}]$ \footnote{Note that corrupt precoders may also be defined as the set of rotations where $|\theta-k \frac{\pi}{4}| \leq  \gamma^{-0.5}$ for $k=1,\ldots,7$; but then, for $k=1,3,5$ and $7$, the lower bound on the coding rate yielding a diversity order of $1.5$ is $\frac{\log_2(3)}{2}$.}. The probability to have such a corrupt precoder is proportional to $\gamma^{-0.5}$. When a precoder is corrupt, then the mutual information is strictly smaller than one in the event of one bad fading gain (see App. \ref{sec: proof lemma toy example} for a formal definition of bad fading). In this case, a coding rate larger than one-half leads to a spectral efficiency that is larger than one, yielding an outage event. The probability of a bad fading gain, which yields an outage event in conjunction with a corrupt precoder, is proportional to $\gamma^{-1}$. The probability of corrupt precoders is unfortunately non-negligible. Lemma \ref{lemma: random prec toy ex} is corroborated by means of numerical simulations, presented in Fig. \ref{fig: outage toy example}. 
\begin{figure}
	\centering
	\includegraphics[width=0.46 \textwidth]{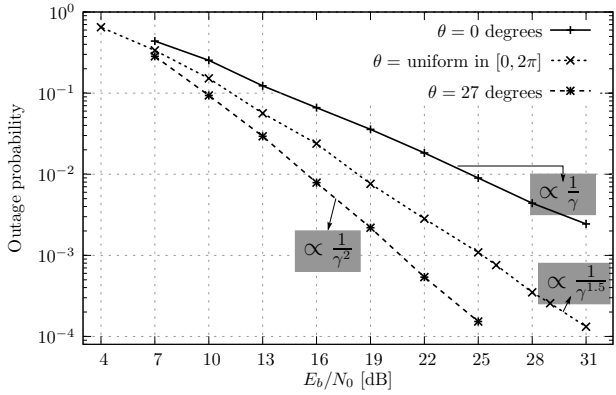}
	\caption{The outage probability of a parallel channel with two states does not achieve full diversity when the input constellation is rotated with a rotation angle that is random (uniform in $[0, 2 \pi]$) but fixed once it is chosen. For reference, the outage probabilities with a fixed rotation (with $\theta=0$ and $27$ degrees, corresponding to a bad and good rotation) are given. Other simulation parameters: $\Omega_z=$BPSK and $R_c=0.9$.}
	\label{fig: outage toy example} 
\end{figure}
When the rotation angle is zero, we have a bad precoder and no transmit diversity is achieved. When the rotation angle is constant and different from $\{k \pi/4\}$, full transmit diversity is achieved. However, when the rotation angle is random, full transmit diversity is not achieved.

\subsection{Random time-varying rotation}
\label{Random time-varying rotation}

The effect of corrupt precoders can be easily reduced by using multiple random precoders, say $N$, during the transmission of one codeword. We leave a formal description for Sec. \ref{sec: TVSOC} but qualitatively describe what happens in the case of the toy example. The probability to have one bad fading gain behaves as $\frac{1}{\gamma}$ (the probability of having two bad fading gains is not dominant, as it behaves as $\frac{1}{\gamma^2}$ and thus automatically leads to full diversity). Given that one fading gain is bad, the mutual information at high SNR converges to one and two for corrupt and non-corrupt precoders, respectively. The mutual information is averaged over these $N$ precoders. We have to consider at most one corrupt precoder as the probability to have a corrupt precoders behaves as $\gamma^{-0.5}$, so that two corrupt precoders in conjunction with a bad fading gain also automatically yield full diversity. In that case, the mutual information for large SNR, averaged over $N$ precoders, converges to $2\frac{N-1}{N}+\frac{1}{N} = 2(1-\frac{1}{2N})$, so that any coding rate smaller than $1-\frac{1}{2 N}$ yields full diversity. 

A similar reasoning is valid for MIMO, which is formalized in the next sections.

\section{Bad and Corrupt precoders for MIMO}
\label{sec: Bad and Corrupt precoders for MIMO}

The definition of bad precoders was given in Def. \ref{def: bad precoders}. Bad precoders decrease the maximal mutual information that can be achieved. Listing the bad precoders is laborious and is fortunately not required for the theoretical analysis of fixed and time-varying space-only codes in the remainder of the paper. For completeness, we illustrate the bad precoders for $\Omega_z = 4-$QAM and a $2 \times 2$ MIMO channel in App. \ref{app: bad precoders for QPSK}.

Bad precoders are defined to yield at least one component of $\mb{s}_t$ with a zero magnitude. The following lemma proves that at most $n_T-1$ components of $\mb{s}_t$ can have a zero magnitude.
\begin{lemma}
	 At most $n_T-1$ elements of $\mb{s}_t$ can have a magnitude equal to zero.
	 \label{lemma: no $n_T$ zero elements}
\end{lemma}
\begin{proof}
	If $n_T$ elements of $\mb{s}_t$ have a zero magnitude, then $\exists~ \mb{z}, \mb{z}^\prime \neq \mb{z}$, satisfying $\mb{x}_t = \mb{x}_t^\prime$, which is impossible, because $V_t$ is a bijection (see Lemma \ref{lemma: U properties} in App. \ref{app: Haar}). 
\end{proof}

As illustrated in the toy example in Sec. \ref{sec: toy example}, the loss of transmit diversity for space-only codes is caused by corrupt precoders, which is an uncountable set of precoders that behave similarly as bad precoders, in the sense that a number of components of $\mb{s}_t$ have a vanishing magnitude with the SNR, instead of being zero for bad precoders. We distinguish between two sets of precoders, $\Set_{c,1}$ and $\Set_{c,2}$, which we both refer to as \textit{corrupt} precoders. The first set is relevant in the case where ${n_R} \geq n_T$. 
\begin{definition}
	We define the set $\Set_{c,1}$ of \textit{corrupt} precoders as the set of precoders so that $\exists~\mb{z}, \mb{z}^\prime \neq \mb{z}$, satisfying $|s_{t,i}| \leq \gamma^{-0.5}, \forall~i>1$. 
	\label{def: corrupt precoders S_c,1}
\end{definition}
For the case that ${n_R}<n_T$, we define corrupt precoders as follows.
\begin{definition}
	We define the set $\Set_{c,2}$ of \textit{corrupt} precoders as the set of precoders so that $\exists~\mb{z}, \mb{z}^{\prime} \neq \mb{z}$, satisfying $|s_{t,i}| \leq \gamma^{-0.5}, i=1, \ldots, {n_R}$. 
	\label{def: corrupt precoders S_c,2}
\end{definition}

Even without determining the structure of these precoders, we can derive their probability of occurrence.
\begin{lemma}
	The probability that a random precoder falls in the set of corrupt precoders is 
\begin{equation}
	\Prob(\Set_{c,1}) \dot{=} \gamma^{-(n_T-1)}, ~~~ \Prob(\Set_{c,2}) \dot{=} \gamma^{-{n_R}}.
\end{equation}	
	for the sets $\Set_{c,1}$ and $\Set_{c,2}$, respectively.
	\label{lemma: Prob Corrupt Precoder}
\end{lemma}
\begin{proof}
	See App. \ref{app: probability of corrupt precoders}.
\end{proof}
In the remainder of the paper, we use the notation $\Set_c$, which refers to $\Set_{c,1}$ or $\Set_{c,2}$, when ${n_R} \geq n_T$ and ${n_R} < n_T$, respectively. As we will see in Sec. \ref{sec: FSOC}, the probability $\Prob(\Set_c)$ is non-negligible and causes the loss of transmit diversity.

\section{Fixed space-only codes}
\label{sec: FSOC}

Consider a unitary but \textit{fixed} space-only precoder $P_t = P$, hence $V_t$ is random (uniformly distributed in the set of all unitary matrices $\mathcal{M}(n_T,n_T)$ (Lemma \ref{lemma: U properties})), but fixed once it is chosen. Because $V_t$ remains constant during the transmission of an outer codeword, we drop the index $t$ in the vectors $\mb{x}_t, \mb{y}_t$ and $\mb{w}_t$, and denote $V^\prime=V^\dagger P$. A new channel $H^\prime = U \Sigma V^\prime$ is formed, which has the same distribution as $H$ because $V^\prime$ has the same distribution as $V^\dagger$. As a consequence, a fixed space-only precoder does not achieve transmit diversity. 

In Sec. \ref{sec: MIMO vs. parallel channels}, the relation between MIMO and parallel channels was given. More specifically, a MIMO channel is equivalent to a parallel channel with random precoding. The probability of having a bad precoder (Def. \ref{def: bad precoders}) is zero, but corrupt precoders cause the diversity loss, which is formalized in the following lemma.

\begin{lemma}
	In a point-to-point flat fading $n_T \times {n_R}$ MIMO channel with a fixed $n_T \times n_T$ precoder $P_t = P$, there exists a coding rate $R_c < 1$ above which the receive diversity is achieved (i.e., the outage SNR exponent is $d_\out = {n_R}$) due to corrupt precoders $V^\prime \in \Set_c$.
	\label{lemma: no full diversity mimo}
\end{lemma}
\begin{proof}
See App. \ref{app: proof lemma no full div mimo}.
\end{proof}
In accordance with the terminology of parallel channels with precoding, the loss of transmit diversity is caused because the random precoder falls too often in the set of corrupt precoders, or more precisely, because the probability to have a corrupt precoder does not converge fast enough to zero. Note that $\Set_{c,1}$ contains precoders with constraints on $n_T-1$ components of $\mb{s}_t$. A larger set, having constraints on less than $n_T-1$ components can be defined as well, but leads to a less tight lower bound on the outage probability when using the same proof techniques as in App. \ref{app: proof lemma no full div mimo}.

Adopting a time-varying precoder $P_t$, uniform in the set of unitary matrices, within the coherence time of the channel (thus, assuming that the MIMO channel remains constant) has no effect on the diversity order for an uncoded scenario with respect to a fixed space-only code $P$. The reason is that the new channel $H^\prime = U \Sigma V_t$ still has the same distribution as $H$, with $V_t = V^\dagger P_t$.

\section{Time-varying space-only codes, achieving full rate and full diversity}
\label{sec: TVSOC}

We propose to reduce the effect of corrupt precoders by averaging the mutual information over $N$ unitary precoders $V_t = V^\dagger P_t$ during the transmission of a codeword. We create $N$ realizations $V_t$ by changing $P_t$ after every $\frac{N_c}{N}$ channel uses. When $N$ is finite, the time-varying space-only code is denoted as \textit{EMI-$N$ code}. In this section, we assume that $N=N_c$ and $N_c \rightarrow \infty$; the corresponding time-varying space-only code is denoted as \textit{EMI code}. The convergence properties of the EMI-$N$ code are discussed in Sec. \ref{sec: Maximal coding rate for EMI-N code}. By letting $P_t$ being uniformly distributed in the set of all unitary matrices $\mathcal{M}(n_T,n_T)$, the matrix $V_t$ is also uniformly distributed in $\mathcal{M}(n_T,n_T)$ (Lemma \ref{lemma: U properties}). Furthermore, by drawing $P_t$ independently at each channel use, $V_t$ changes independently at each channel use. 

Note that the bottom $n_T-{n_R}$ elements of $\mb{z}$ are zero when ${n_R} < n_T$, see Sec. \ref{Full rate space-time and space-only coding}.

\begin{theorem}
	In a point-to-point flat fading $n_T \times {n_R}$ MIMO channel, using a coding rate $R_c < 1$ and using $n_T \times n_T$ precoders $P_t$, randomly generated for each channel use, being uniformly distributed in the set of all unitary matrices $\mathcal{M}(n_T,n_T)$, full diversity is achievable.
	\label{prop: full diversity mimo with distr rotation}
\end{theorem}
\begin{proof}
See App. \ref{app: full diversity mimo with distr rotation}.
\end{proof}
Note that in (\ref{sample mean in app}) in App. \ref{app: full diversity mimo with distr rotation}, we used the fact that the temporal mean equals the sample mean, i.e.,
\begin{equation}
	\E_t \left[  I(\mb{x}_t; \mb{y}_t | \Sigma, V_t) \right] = \E \left[  I(\mb{x}_t; \mb{y}_t | \Sigma, V_t) \right], \label{sample mean to temporal mean}
\end{equation}
where the expectation at the right side in (\ref{sample mean to temporal mean}) is over $V_t$, uniformly distributed in the set of unitary matrices $\mathcal{M}(n_T, n_T)$, for  a fixed $t$. This is valid for the EMI code. 

Theorem \ref{prop: full diversity mimo with distr rotation} corroborates experimental results in the literature. For example, in \cite{hiroike1992ceo, ma2005stm}, the precoder $P_t$ was made variable by multiplying an initial precoder $P_0$ by a diagonal matrix $A_t$ including $e^{j 2 \pi \theta_i(t)}$ on the $i$-th diagonal element, where $\theta_i(t)$ varies each channel use. The overall precoder $P_0 A_t$ is unitary and time-varying\footnote{Note however that $P_0 A_t$ is not uniform in the set of unitary matrices.}. Because the multiplication with $A_t$ corresponds to adding a time-varying phase to the complex baseband signal at each transmit antenna, the scheme was referred to as phase sweeping. Simulation results in \cite{hiroike1992ceo, ma2005stm} suggest that full diversity is achieved in the presence of an error-correcting code with a particular coding rate. Theorem \ref{prop: full diversity mimo with distr rotation} now proves that full diversity is actually achievable (through the outage probability) for codes with any coding rate smaller than one, as long as $P_t$ is uniform in the set of unitary matrices.

\section{Numerical results}
\label{Numerical results}

In this section, we corroborate Theorem \ref{prop: full diversity mimo with distr rotation} by numerically determining the outage probability for the EMI-$N$ code, which are also compared with the outage probabilities of approximately universal space-time codes. Despite being an unfair comparison taking into account the detection complexity, it allows to assess the loss of coding gain as a price for the reduced detection complexity associated with the EMI code. The outage probabilities are achievable lower bounds of the word error rate (WER) of practical coded systems. Therefore, it is useful to compare the outage probabilities with the WER of a coding scheme having the EMI code as inner code and an error-correcting code as outer code. 

In Sec. \ref{sec: LDPC code optimization for the EMI code}, we discuss the optimization of LDPC codes as outer code in such a coding scheme. Next, the convergence properties of the EMI-$N$ code to the EMI code are discussed in Sec. \ref{sec: Maximal coding rate for EMI-N code}. In the last two subsections, we present the numerical results for the $n_T \times {n_R}$ MIMO channel with ${n_R} \geq n_T$ and ${n_R} < n_T$, respectively. 

The outage probabilities and WERs of the LDPC coded modulations are determined by means of Monte Carlo simulations. If the decoding complexity of exhaustive ML-detection of the STCs was too high, then sphere decoding is performed \cite{viterbo1999aul, Agrell2002cps, Boutros2003sis}. More specifically, soft output sphere decoding is performed \cite{Boutros2003sis}, so that the Tanner graph of the LDPC code gets a soft input from the detectors of the STC. Iterative decoding and detection is performed, where after every 10 LDPC decoding iterations, a new detection of the STC is performed. The LDPC code is decoded by means of the sum-product algorithm on the Tanner graph. The total number of decoding iterations is limited by $100$.  

The computation of (\ref{eq: mut info discrete alphabet MIMO form 1}) in the numerical calculation of the outage probability is in some cases (e.g. for approximately universal STCs, large constellation size or too many transmit antennas) time consuming because the constellation $\Omega_{\mb{x}}$ might be very large. Therefore, the inner sum over $\mb{x}_t^\prime \in \Omega_{\mb{x}}$ in (\ref{eq: mut info discrete alphabet MIMO form 1}) is simplified by using a sphere decoder. The sphere decoder outputs a list of maximum likelihood metrics $d^2(\mb{y}_t, \sqrt{\gamma} \Sigma \mb{x}_t^\prime)$ for all closest constellation points to the constellation point with the highest likelihood \cite{Agrell2002cps}. The list size is limited to $1000$. 

\subsection{LDPC code optimization for the EMI code}
\label{sec: LDPC code optimization for the EMI code}

We focus on binary LDPC codes $\mathcal{C}[N_b,K_b]_2$ with dimension $K_b$, so that the coding rate $R_c = \frac{K_b}{N_b}$. Irregularity is introduced through the standard bit and check node degree distributions, characterized by the polynomials $\lambda(x)$ and $\rho(x)$, here from an edge perspective \cite{richardson2001dca}. 

It is worth citing \cite{ten2004dol, yue2005ooi, zheng2006lcm} and references therein, where the degree distributions of LDPC codes have been optimized for Gaussian and ergodic MIMO channels. However, optimizing these degree distributions for a non-ergodic MIMO channel (hence, where the random channel remains constant during the whole codeword) has not yet been performed to the best knowledge of the authors. In Sec. \ref{sec: MIMO vs. parallel channels}, we have explained that the non-ergodic MIMO channel is equivalent to a parallel channel with a random precoder. To the authors' best knowledge, only \cite{duy2011pfw} presented a tool to optimize the degree distributions of an LDPC code for parallel Rayleigh faded channels with precoding. We therefore utilized the techniques from \cite{duy2011pfw} to generate the degree distributions of the LDPC code. 

We briefly describe the techniques developed in \cite{duy2011pfw}, where coded modulations (including the precoding matrix, the mapping function and the error-correcting code) were optimized to yield a WER closely approaching the outage probability. The off-line optimization, done using a geometric approach, was limited to at most $B+1$ times the effort for Gaussian channels, where $B$ is the number of parallel channels, here $\min(n_T,{n_R})$. The geometric approach involves optimizing the coded modulation, e.g. via EXIT charts, in $B+1$ well chosen points in the $B$-dimensional fading plane. More specifically, in $B$ fading points on the $B$ axes of the fading plane close to the outage boundary, plus one fading point on the ergodic line ($\alpha_1 = \ldots = \alpha_B$). The outage boundary limits the set of fading gain vectors for which the instantaneous mutual information is smaller than the spectral efficiency. More details on the outage boundary and the optimization of the outage probability through a proper selection of the precoding matrix and the constellation $\Omega_{\mb{z}}$ can be found in \cite{duy2011pfo}. 

The main difference between the current channel equation (see Eq. \ref{eq: new channel eq 2}) and the channel model in \cite{duy2011pfw} is that the precoder $V_t$ is random. Hence, the instantaneous mutual information $I(\mb{x}_t;\mb{y}_t|\Sigma,V_t)$ does not only depend on the instantaneous realization of $\Sigma$, but also on the precoder realization $V_t$ (Fig. \ref{fig: capa variation on outage boundary} shows the variation of the instantaneous mutual information with $V_t$). Because of the dependence on $V_t$, an outage boundary for the singular values $\{\sigma_i, i = 1, \ldots, \min({n_R},n_T)\}$ cannot be defined based on $I(\mb{x}_t;\mb{y}_t|\Sigma,V_t)$. However, when the EMI-code is used, the mutual information is averaged over all unitary precoders in $\mathcal{M}(n_T,n_T)$, so that the mutual information $\E_t[I(\mb{x}_t; \mb{y}_t | \Sigma, V_t)]$ only depends $\Sigma$. We define the corresponding outage boundary as the set of singular values $\{\sigma_i, i = 1, \ldots, \min({n_R},n_T)\}$ which yield $\E_t[I(\mb{x}_t; \mb{y}_t | \Sigma, V_t)] = R$. In practice, we approximate $\E_t[I(\mb{x}_t; \mb{y}_t | \Sigma, V_t)]$ by averaging $I(\mb{x}_t; \mb{y}_t | \Sigma, V_t)$ over a large number (e.g. $100$) of randomly generated matrices $V_t$. Fig. \ref{fig: outage boundary 2x2 MIMO QPSK R=3.6} illustrates such an outage boundary. 
\begin{figure*}
\centering
\begin{subfigure}[b]{.35\linewidth}
   \centering \includegraphics[width=0.9 \textwidth, angle=-90]{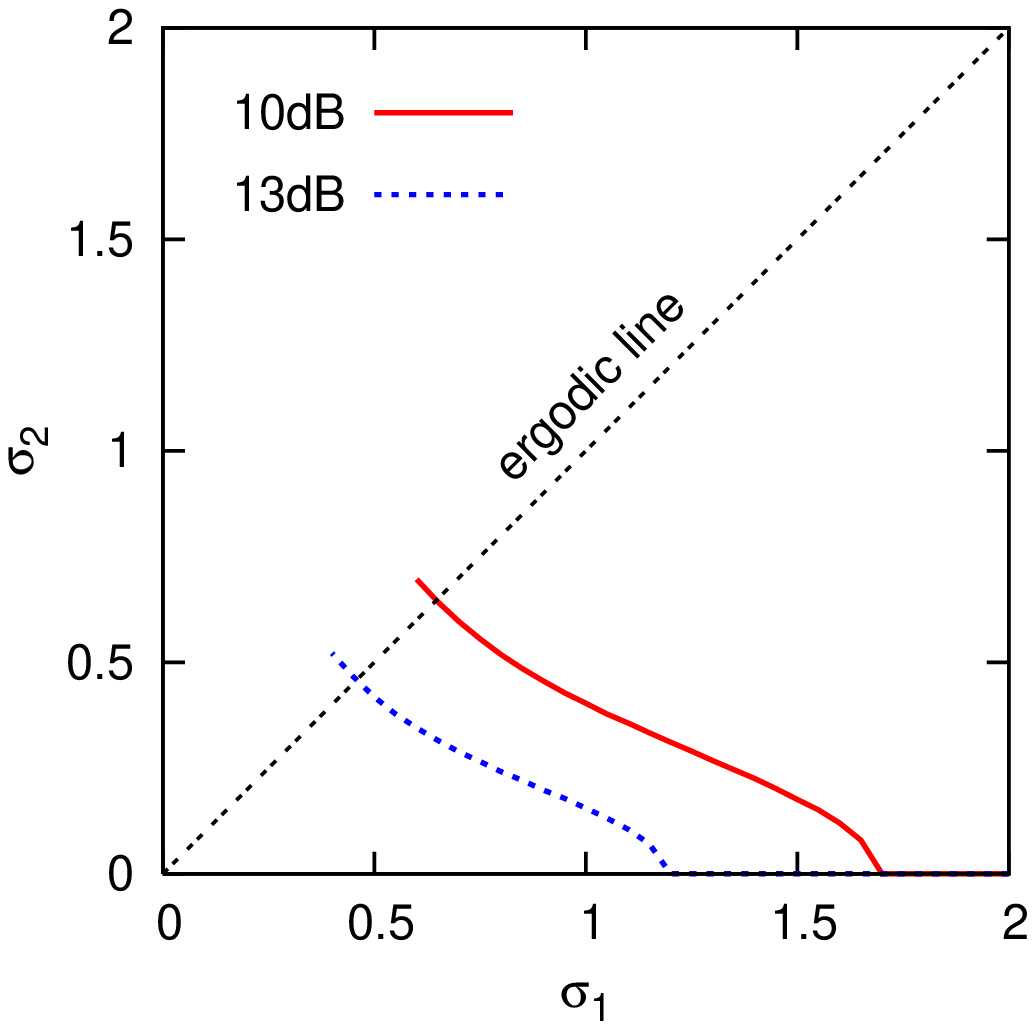}
	\caption{The outage boundary is shown at two SNR values: $E_b/N_0 = 10$ dB and $E_b/N_0 = 13$ dB. Only half of the plane is used, as we consider the \textit{ordered} singular values, $\sigma_1 \geq \sigma_2$.}
	\label{fig: outage boundary 2x2 MIMO QPSK R=3.6}
\end{subfigure}
\quad
\begin{subfigure}[b]{.55\linewidth}
   \centering \includegraphics[width=1.0 \textwidth]{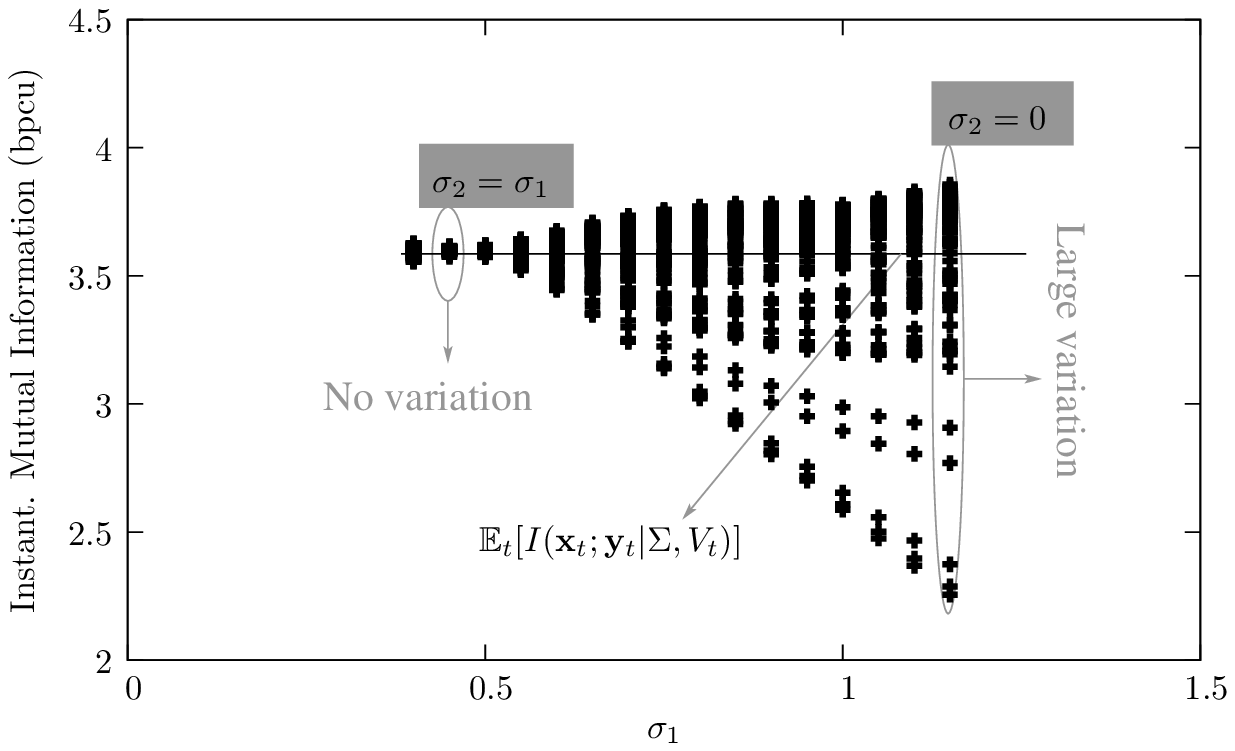}
   \caption{The mutual information $I(\mb{x}_t; \mb{y}_t | \Sigma, V_t)$ largely varies with $V_t$ while $(\sigma_1, \sigma_2)$ is on the outage boundary, especially when $\sigma_1 \gg \sigma_2$. For a given $\Sigma$, the values of $100$ mutual informations $I(\mb{x}_t; \mb{y}_t | \Sigma, V_t)$ are displayed, each corresponding to a realization of $V_t$. The adopted SNR is $E_b/N_0 = 13$ dB.}
   \label{fig: capa variation on outage boundary}
\end{subfigure}
	\caption{At the left side, the outage boundary is shown for $2\times 2$ MIMO with $\Omega_z=4$-QAM and $R=3.6$ bpcu, which is the set of singular values $(\sigma_1, \sigma_2)$ where $\E_t[I(\mb{x}_t; \mb{y}_t | \Sigma, V_t)]$ is equal to $R$. We approximated $\E_t[I(\mb{x}_t; \mb{y}_t | \Sigma, V_t)]$ by averaging $I(\mb{x}_t; \mb{y}_t | \Sigma, V_t)$ over $100$ randomly generated matrices $V_t$. The variation of $I(\mb{x}_t; \mb{y}_t | \Sigma, V_t)$ with $V_t$ is illustrated at the right side, where $(\sigma_1, \sigma_2)$ are located on the outage boundary (the outage boundary is a one-dimensional line, so that only one parameter, e.g. $\sigma_1$, is sufficient to characterize the location on this line).}
	\label{fig: outage bound and variation of mut info with V_t}
\end{figure*}
As a consequence, the same method as in \cite{duy2011pfw} can be applied to generate LDPC code degree distributions when using the EMI inner code. 

\subsection{Maximal coding rate for the EMI-$N$ code}
\label{sec: Maximal coding rate for EMI-N code}

Theorem \ref{prop: full diversity mimo with distr rotation} proved that full diversity is achieved by the EMI code for any coding rate $R_c < 1$. When $P_t$ changes a finite number of times $N$ during a codeword transmission, (\ref{sample mean to temporal mean}) and thus the proof in App. \ref{app: full diversity mimo with distr rotation}, are not valid. The maximum coding rate yielding full diversity may be smaller than for the EMI code. In this section, we provide an upper bound on the maximal coding rate yielding full diversity for the EMI-$N$ code, together with supporting numerical results. 

As the probability of occurrence of corrupt precoders is strictly larger then $\gamma^{-nr}$, the occurrence of corrupt precoders amongst the $N$ precoders observed during the transmission of a codeword has to be considered. To derive the maximal coding rate yielding full-diversity, the maximal mutual information that can be achieved between input and output, when the input is transformed by a corrupt precoder, needs to be determined. Then, through a similar formula as Eq. (\ref{eq: maximal mutual info full div emi}), the maximal coding rate could be determined. 

For example, consider the different classes of bad precoders for the $2 \times 2$ MIMO channel and $\Omega_z=4$-QAM, which are given in Eq. (\ref{eq: bad precoders 4-QAM}). Each type of bad precoder corresponds to a maximum achievable mutual information $I(\mb{x}_t; \mb{y}_t | \Sigma, V_t)$, which for the precoders from (\ref{eq: bad precoders 4-QAM}) are given by $\{2, 2, \log_2 9, \log_2 12, \log_2 12\}$. This can be verified analytically or by means of simulation. Next, different subsets of the set of corrupt precoders will have a maximal mutual information that converges to the maximal mutual information of one of the bad precoders for increasing SNR. Each subset will have to be defined more precisely than in Defs. \ref{def: corrupt precoders S_c,1} and \ref{def: corrupt precoders S_c,2}, i.e., some subsets will have more than one pair of $\mb{z}, \mb{z}^\prime$ yielding $|s_{t,i}| \leq \gamma^{-0.5}$ for some $i$. Therefore, some subsets might have a smaller probability of occurrence than what is given in Lemma \ref{lemma: Prob Corrupt Precoder}. 

A detailed study of this behaviour is outside the scope of this work. However, we can provide an upper bound on the maximal coding rate yielding full diversity, by assuming an optimistic case. Instead of averaging $I(\mb{x}_t; \mb{y}_t | \Sigma, V_t)$ over $V_t = V^\dagger P_t$ by changing $P_t$, we can consider the average of $I(\mb{x}_t; \mb{y}_t | H)$ over $N$ realizations of $H$, corresponding to a block fading MIMO channel with $N$ states. The achievable diversity order for a block fading MIMO channel with $N$ blocks is upper bounded by \cite{gresset2008stc}
\begin{equation}
	d_{\textrm{up}}= {n_R} (\lfloor N n_T (1-R_c) \rfloor+1). 
\end{equation}
Hence, by putting $d_{\textrm{up}} \geq nr$, we get after some calculus that
\begin{equation}
	R_{c,\textrm{max}} \leq 1 - \frac{1}{N} + \frac{1}{nN}. 
\label{eq: Singleton upper bound EMI-N}
\end{equation}
Further work must verify whether this bound is tight for the EMI-$N$ code. 


%

The theoretical analysis is laborious, but conjectures can be made from the numerical simulations by comparing the diversity order of the EMI-$N$ code with that of the EMI code, which achieves full diversity (Fig. \ref{fig: Outage 2Rot and 10Rot}). When using $N=2$, it may be conjectured that the bound in Eq. (\ref{eq: Singleton upper bound EMI-N}) is achieved as full diversity is still achieved for a coding rate $R_c=0.75$. When using $N=10$, full diversity is still achieved for a coding rate $R_c=0.9$. However, at a coding rate $R_c=0.95$, which corresponds to the bound in Eq. (\ref{eq: Singleton upper bound EMI-N}), full diversity is not achieved.

\begin{figure*}
\begin{subfigure}[b]{.47\linewidth}
   \centering \includegraphics[width=1.0 \textwidth]{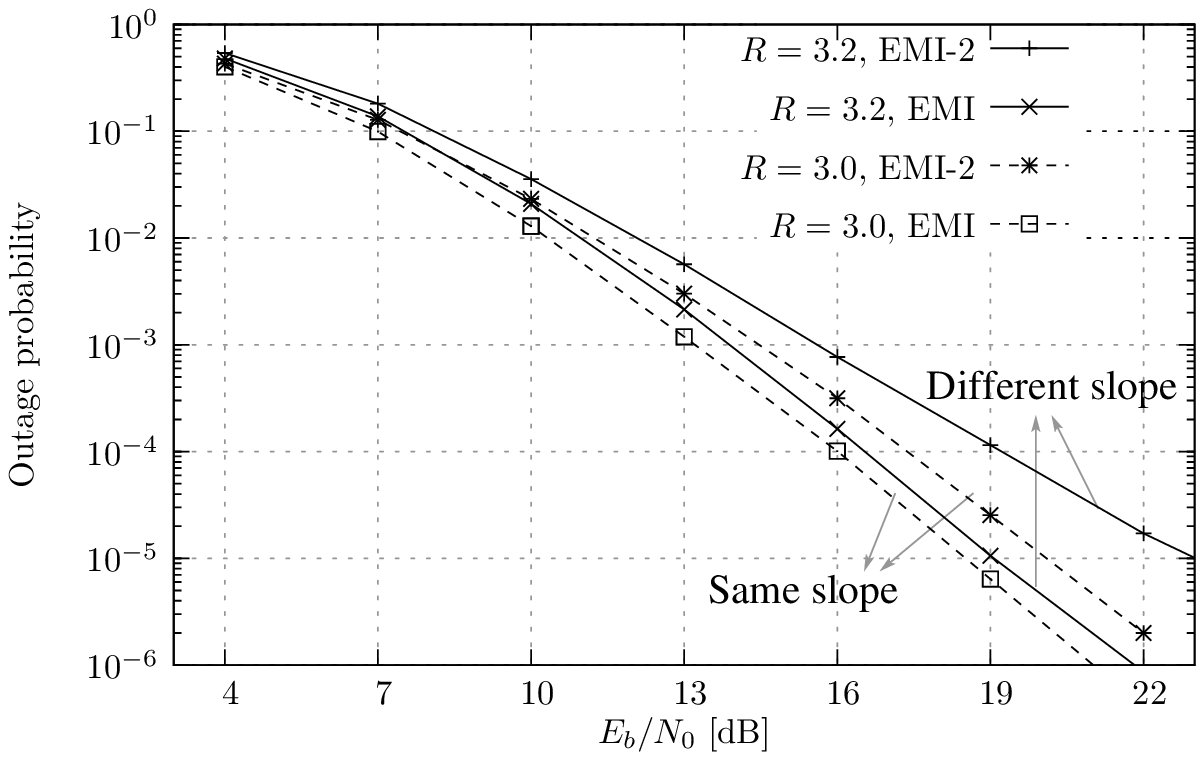}
	\caption{The outage probabilities of the EMI-$2$ for $R=\{3,~3.2\}$ bpcu are shown. It may be conjectured that full diversity is achieved for rates not larger than $R=3.0$ bpcu, which corresponds to $R_c=0.75$.}
	\label{fig: Outage 2Rot nt=nr=2 QPSK}
\end{subfigure}
\quad
\begin{subfigure}[b]{.47\linewidth}
	\centering \includegraphics[width=1.0 \textwidth]{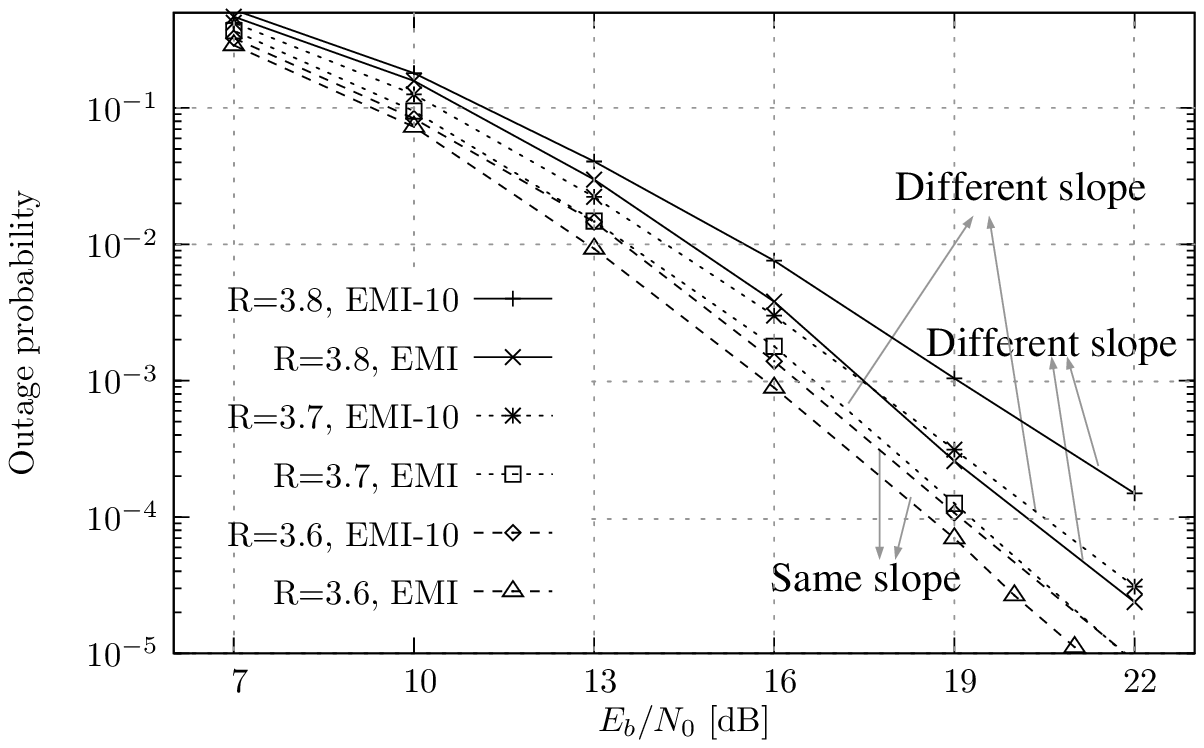}
	\caption{The outage probabilities of the EMI-$10$ for $R=\{3.6,~3.7,~ 3.8\}$ bpcu are shown. It may be conjectured that full diversity is achieved for rates $R \leq 3.6$ bpcu, which corresponds to $R_c=0.9$.}
	\label{fig: Outage 10Rot nt=nr=2 QPSK}
\end{subfigure}
	\caption{It can be numerically verified whether the EMI-$N$ code (EMI-$2$ at the left side and EMI-$10$ at the right side) achieves full diversity by comparing its outage probability to that of the EMI-code, which always achieves full diversity. Simulation parameters: $2\times 2$ MIMO with $\Omega_z=4$-QAM and different rates.}
	\label{fig: Outage 2Rot and 10Rot}
\end{figure*}

Because it is difficult to design very high coding rate error-correcting codes (see \cite{yang2004doe} and references therein), coding rates larger than $R_c=0.9$ almost never occur in practice. Therefore, we only consider the EMI-$10$ code, achieving full diversity at coding rates smaller than or equal to $R_c=0.9$, in the remainder of the paper, unless mentioned otherwise.

\subsection{Numerical results for the $n_T \times {n_R}$ MIMO, ${n_R} \geq n_T$}

When ${n_R} \geq n_T$, a transformation of $n_T$ independent symbols are transmitted at each channel use. We compare the outage and error rate performance of the EMI-$10$ code with STCs from the literature that are optimal in terms of uncoded error rate (e.g., \cite{belfiore2005tgc, elia2007pstc} and \cite{boutros2009tap}) or optimal in terms of error rate assuming a genie condition (e.g., \cite{boutros2009tap} and \cite{gresset2008stc}). We also include the error rate performance without precoding (which is equivalent to a fixed space-only code) for reference. 

First, let us consider the $2 \times 2$ MIMO channel and $\Omega_z=4$-QAM. Fig. \ref{fig: Code1_N=5760_Detmult=10_SpDec_ldpc_MIMO_nt=nr=2_R=3.6_QPSK.eps} compares the performances of the Golden code \cite{belfiore2005tgc} with the EMI-$10$ code and the case without precoding. The performance of other STCs, such as \cite{boutros2009tap, gresset2008stc} were also tested, but their error rates are slightly larger than that of the Golden code, so that they are not displayed for clarity. The LDPC code of $R_c=0.9$ is a regular $(3,30)$ LDPC code, which showed the best performance among other tested LDPC codes of the same coding rate. The code length of the LDPC code is $N_b = 5760$.
\begin{figure}[!ht]
	\centering
	\includegraphics[width=0.49 \textwidth]{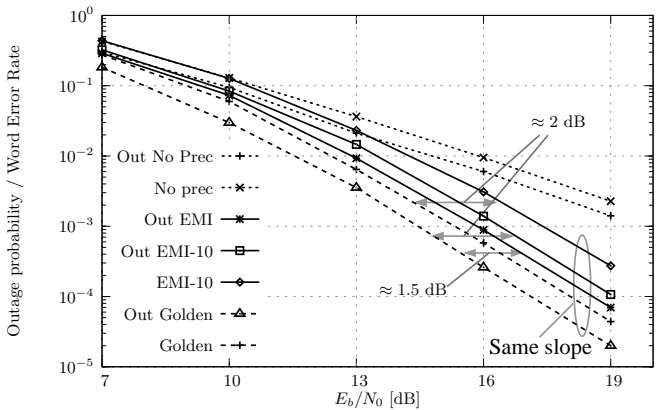}
	\caption{For a coding rate $R_c=0.9$ (i.e., $R=3.6$ bpcu) and a $2 \times 2$ MIMO channel, the outage probabilities of the EMI and EMI-$10$ code achieve full diversity and perform at $1.5$ and $2$ dB from the outage probability of the Golden code. ``Out'' stands for outage probability, the other curves represent the WER of a practical system with an error-correcting code.}
	\label{fig: Code1_N=5760_Detmult=10_SpDec_ldpc_MIMO_nt=nr=2_R=3.6_QPSK.eps}
\end{figure}
As theoretically proved in Theorem \ref{prop: full diversity mimo with distr rotation}, full diversity is achieved by the EMI code. The focus of this paper was on full rate, full diversity, and low decoding complexity. Of course, the significant reduction in complexity comes at the expense of a loss in coding gain (the EMI and EMI-$10$ code respectively require $1.5$ and $2$ dB more energy to have the same error rate performance than the Golden code). However, as shown in Figs. \ref{fig: Outages_LDPC_Detmult=10_MIMO_nt=nr=2_R=3.0_qpsk.eps} and \ref{fig: Outage_LDPC_Detmult=10_MIMO_nt=nr=2_R=2.2_qpsk.eps}, this loss in coding gain vanishes when the coding rate decreases. More spefically, the loss in coding gain is around $0.5$ dB when $R_c=0.75$ and $0$ dB when $R_c=0.55$. In Figs. \ref{fig: Outages_LDPC_Detmult=10_MIMO_nt=nr=2_R=3.0_qpsk.eps} and \ref{fig: Outage_LDPC_Detmult=10_MIMO_nt=nr=2_R=2.2_qpsk.eps}, the outage probability of the EMI and EMI-$10$ code coincide, so that the EMI code is not shown for clarity. Similarly, the error rates of the Golden code and the STCs from \cite{boutros2009tap, gresset2008stc} practically coincide. The degree distributions of the LDPC code were obtained as explained in Sec. \ref{sec: LDPC code optimization for the EMI code} and are given by
\[ \lambda(x) = 0.231 x + 0.543x^2 + 0.226 x^{19}; ~~~ \rho(x) = x^{12}\]
for $R_c=0.75$ and by
\[ \lambda(x) = 0.215 x + 0.465x^2 + 0.040x^{19} +0.280 x^{20}; ~~~ \rho(x) = x^{7}\]
for $R_c=0.55$. The code length is $N_b=5760$.

\begin{figure}[!ht]
	\centering
	\includegraphics[width=0.49 \textwidth]{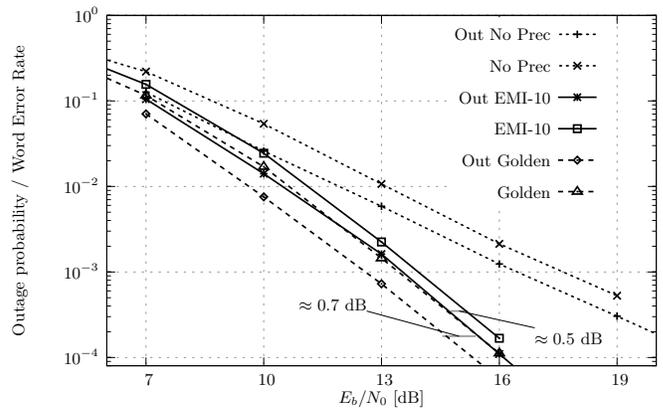}
	\caption{For a coding rate $R_c=0.75$ (i.e., $R=3.0$ bpcu), the horizontal SNR-gap between the outage probabilities of the EMI-$10$ code and the Golden code reduced from $2$ dB (for $R_c=0.9$) to $0.5$ dB. Other simulation parameters: $2 \times 2$ MIMO, ``Out'' stands for outage probability, the other curves represent the LDPC code WER.}
	\label{fig: Outages_LDPC_Detmult=10_MIMO_nt=nr=2_R=3.0_qpsk.eps}
\end{figure}

\begin{figure}[!ht]
	\centering
	\includegraphics[width=0.49 \textwidth]{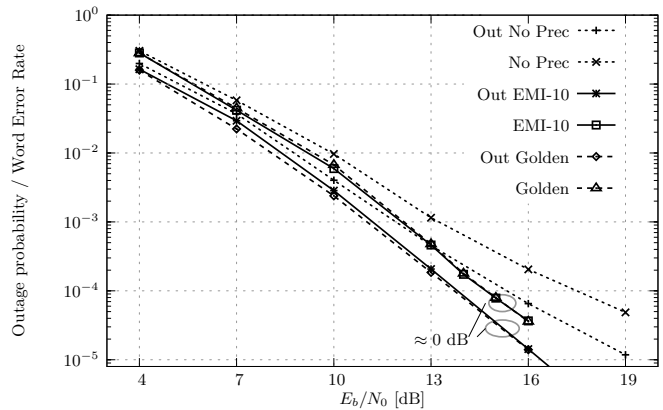}
	\caption{For a coding rate $R_c=0.55$ (i.e., $R=2.2$ bpcu), the horizontal SNR-gap between the outage probabilities of the EMI-$10$ code and the Golden code has vanished. Other simulation parameters: $2 \times 2$ MIMO, ``Out'' stands for outage probability, the other curves represent the LDPC code WER.}
	\label{fig: Outage_LDPC_Detmult=10_MIMO_nt=nr=2_R=2.2_qpsk.eps}
\end{figure}

Next, we consider the $2 \times 3$ and $3 \times 3$ MIMO channels using $\Omega_z=4$-QAM. Figs. \ref{fig: Code1_N=5760_Detmult=10_SpDec_ldpc_MIMO_nt=2_nr=3_R=3.6_QPSK.eps} and \ref{fig: Code1_N=5760_Detmult=10_SpDec_ldpc_MIMO_nt=nr=3_R=5.4_QPSK.eps} compare the performances of approximately universal STCs (the Golden code \cite{belfiore2005tgc} for $n_T=2$ and the Perfect STC \cite{elia2007pstc} for $n_T=3$) with the EMI-$10$ code and the case without precoding. Other STCs, such as \cite{boutros2009tap, gresset2008stc} for $n_T=2$, have a similar or slightly degraded error rate performance, so that they are not displayed for clarity. The LDPC code of $R_c=0.9$ is again a regular $(3,30)$ LDPC code and the code length of the LDPC code is $N_b = 5760$.
\begin{figure}[!ht]
	\centering
	\includegraphics[width=0.49 \textwidth]{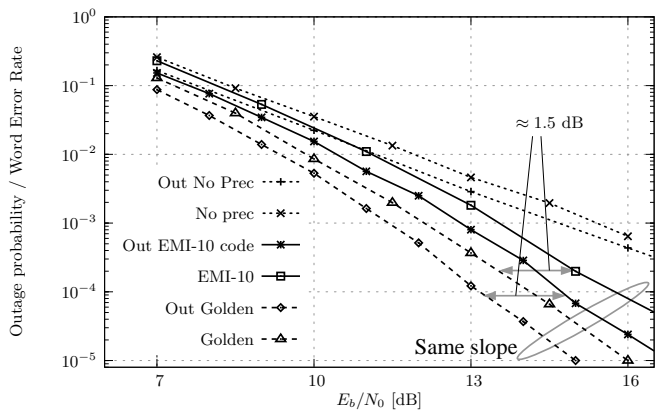}
	\caption{The EMI-$10$ code also achieves full diversity for $R_c=0.9$ (i.e., $R=3.6$ bpcu) and the $2 \times 3$ MIMO channel, performing at approximately $1.5$ dB from the Golden code. ``Out'' stands for outage probability, the other curves represent the LDPC code WER.}
	\label{fig: Code1_N=5760_Detmult=10_SpDec_ldpc_MIMO_nt=2_nr=3_R=3.6_QPSK.eps}
\end{figure}
\begin{figure}[!ht]
	\centering
	\includegraphics[width=0.49 \textwidth]{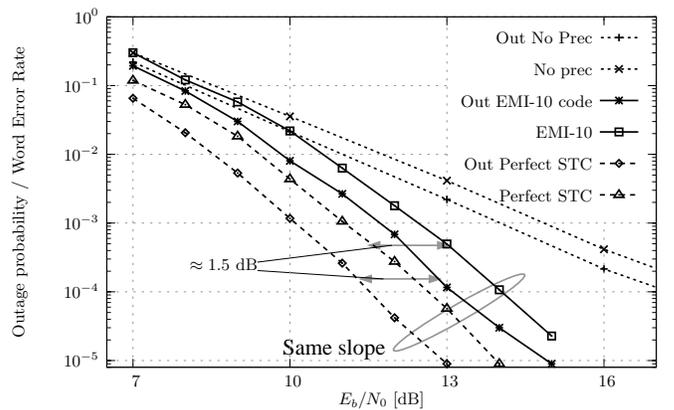}
	\caption{Also for the $3 \times 3$ MIMO channel and $R_c=0.9$ (i.e., $R=5.4$ bpcu), full-diversity is achieved by the EMI-$10$ code, again performing at approximately $1.5$ dB from the approximately universal code, denoted as perfect STC. ``Out'' stands for outage probability, the other curves represent the LDPC code WER.}
	\label{fig: Code1_N=5760_Detmult=10_SpDec_ldpc_MIMO_nt=nr=3_R=5.4_QPSK.eps}
\end{figure}
Full diversity is achieved by the EMI-$10$ code and as for the $2 \times 2$ MIMO channel, the significant reduction in complexity comes at the expense of a loss in coding gain when the coding rate is close to one. When $R_c=0.9$, this loss is approximately $1.5$ dB. 

Finally, we consider a larger constellation size for the $2 \times 2$ MIMO channel, using $\Omega_z=16$-QAM (Fig. \ref{fig: Code1_N=5760_Detmult=10_SpDec_ldpc_MIMO_nt=nr=2_R=7.2_16qam.eps}). The other parameters, such as the degree distributions, coding rate and the code length are the same. The performance of the EMI-$10$ code is compared with those of the Golden code, the cyclotomic code \cite{gresset2008stc} and Aladdin-Pythagoras code \cite{boutros2009tap}. The loss in coding gain of the EMI-code is not larger than $1.5$ dB.
\begin{figure}[!ht]
	\centering
	\includegraphics[width=0.49 \textwidth]{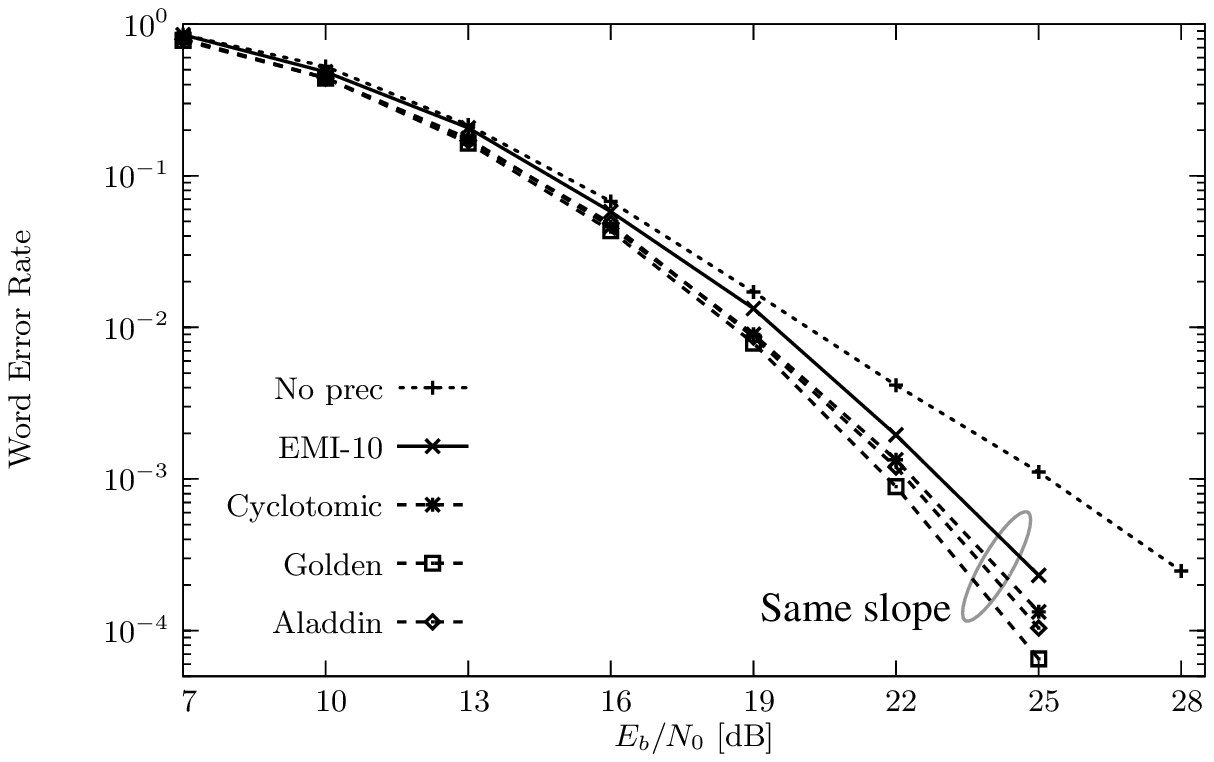}
	\caption{Also for higher constellation sizes (here $\Omega_z=16$-QAM), the EMI-$10$ code achieves full diversity in terms of WER. Other simulation parameters: $2 \times 2$ MIMO, $R_c=0.9$ ($R=7.2$ bpcu).}
	\label{fig: Code1_N=5760_Detmult=10_SpDec_ldpc_MIMO_nt=nr=2_R=7.2_16qam.eps}
\end{figure}

\subsection{Numerical results for the $n_T \times {n_R}$ MIMO, ${n_R} < n_T$}

When ${n_R} < n_T$, a transformation of, on average, ${n_R}$ independent symbols are transmitted at each channel use. For the $2 \times 1$ MIMO channel, the Alamouti code \cite{Alamouti1998ast} is well known to be optimal in many ways, e.g. in terms of uncoded error rate. The loss of coding gain of the EMI-$N$ code is quite large (e.g., more than $3$ dB for $R_c=0.9$) when the coding rate is close to one, but decreases with the coding rate (e.g., around $1.5$ dB for $R_c=0.75$), see Fig. \ref{fig: Outage 2x1 EMI vs Alamouti}.
\begin{figure*}
\begin{subfigure}[b]{0.48\linewidth}
   \centering \includegraphics[width=1.0 \textwidth]{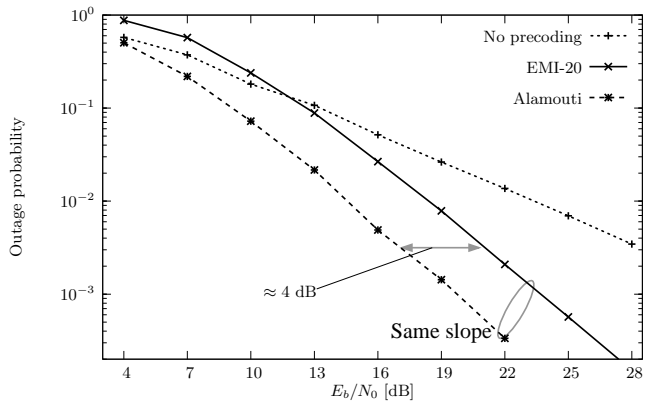}
	\caption{For $R_c=0.9$ (i.e., $R=1.8$ bpcu), the outage probability of the EMI-$20$ code performs at approximately $4$ dB from the Alamouti code.}
	\label{fig: Outage 2Rot nt=nr=2 QPSK}
\end{subfigure}
\qquad
\begin{subfigure}[b]{.48\linewidth}
	\centering  \includegraphics[width=1.0 \textwidth]{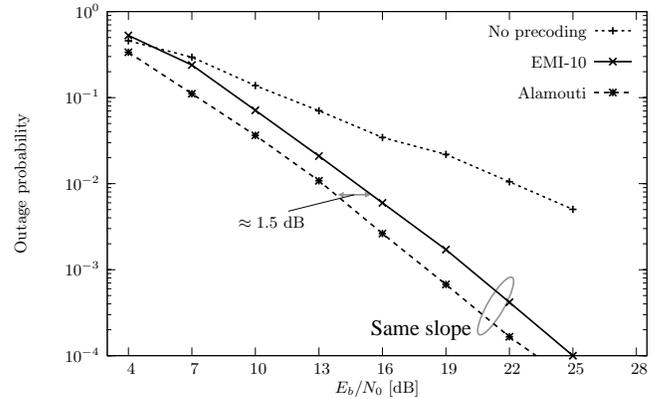}
	\caption{The horizontal SNR-gap with the Alamouti code reduced to approximately $1.5$ dB for $R=1.5$ bpcu ($R_c = 0.75$).}
	\label{fig: Outage 10Rot nt=nr=2 QPSK}
\end{subfigure}	
	\caption{The outage probability of the EMI-$N$ code is compared with the Alamouti STC for the $2\times 1$ MIMO channel with $\Omega_z=4$-QAM and rates $1.8$ bpcu (left) and $1.5$ bpcu (right). Full diversity is achieved but the loss of coding gain is large.}
	\label{fig: Outage 2x1 EMI vs Alamouti}
\end{figure*}
It may be conjectured that the EMI-$N$ code can not compete with the Alamouti code for the $2 \times 1$ MIMO channel, given that the Alamouti code also has a low detection complexity. 

For completeness, we also show the outage probability of the EMI-$N$ code for the $3 \times 1$ MIMO channel (Fig. \ref{fig: Outages_MIMO_nt=3_nr=1_rate=1.8_qpsk.eps}), which corroborates Theorem \ref{prop: full diversity mimo with distr rotation}, i.e., full diversity is achieved, also when ${n_R} <n_T$.
\begin{figure}[!ht]
	\centering
	\includegraphics[width=0.48 \textwidth]{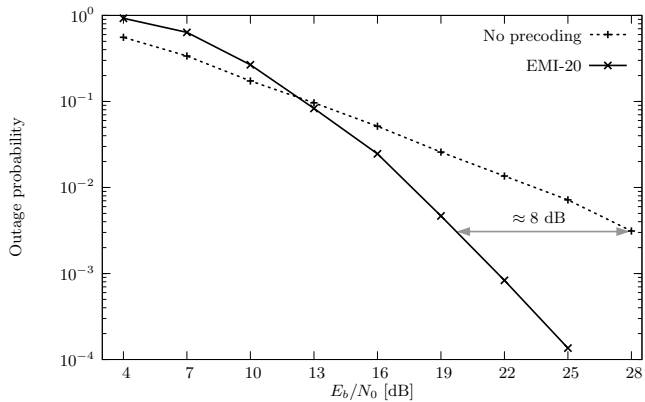}
	\caption{We show that the outage probability of the EMI-$20$ code achieves full diversity on the $3 \times 1$ MIMO channel for $\Omega_z=4$-QAM and $R_c=0.9$ (i.e., $R=1.8$ bpcu).}
	\label{fig: Outages_MIMO_nt=3_nr=1_rate=1.8_qpsk.eps}
\end{figure}

\section{Conclusion}

For unknown channel state information at the transmitter, we proposed a new coded modulation paradigm: outer error-correction code with space-only but time-varying precoding. Full diversity is proved in terms of the outage probability and verified by means of numerical simulations. The proofs are based on the relation between MIMO and parallel channels. For most coding rates, using ten space-only precoders per codeword is sufficient and the error rate performance of LDPC coded modulations approaches the outage probability. The detection complexity of space-only codes is significantly smaller than for approximately universal space-time codes, coming at a price of coding gain loss for coding rates $R_c$ close to one (e.g. $1.5$ dB for $R_c=0.9$), which fortunately vanishes with decreasing coding rate.

\section*{Acknowledgement}

Dieter Duyck thanks Dr. Gareth Amery from the Mathematical Sciences department of the University of Kwazulu Natal, for the interesting discussions on multivariate statistical theory. Dieter Duyck also thank the University of Kwazulu Natal, where the work leading to this paper has taken place.

\appendix
\subsection{Random Unitary Matrices}
\label{app: Haar}

When a unitary matrix is uniformly distributed in the set of all unitary matrices with a particular dimension, then formally, it is uniformly distributed in the Stiefel manifold with respect to the Haar measure. This section is not a part of our contribution, but we briefly summarize the concepts of Stiefel manifold and Haar distribution. For a thorough description, we refer to \cite{choquet1982ama}, \cite{muirhead}.

A unitary matrix $V \in \C^{n_T \times n_T}$ belongs to the Stiefel manifold $\mathcal{M}_{n_T,n_T}$, which is the set of all unitary matrices in $\C^{n_T \times n_T}$ \cite[Sec. III.A]{choquet1982ama} \cite[Sec. 2.1.4]{muirhead} \cite[Sec. 3.2]{edelman1989eac} \cite[Sec. 4.6]{edelman2005rmt}. More specifically, $\mathcal{M}_{n_T,n_T}$ is a subset with $n_T^2$ real dimensions of the inner product space $\C^{n_T \times n_T}$ which has $2n^2$ real dimensions. 

In order to consider a probability density function (pdf) for $V$, a means for integration must be defined first, as the integral of the pdf over the considered space, here the Stiefel manifold, equals one. In the Stiefel manifold, Riemann integration is not applicable so that Lebesgue integration, which is a generalization of Riemann integration to more general spaces, is required \cite[Sec. I.D]{choquet1982ama}. A Riemann integral is the sum of the volumes of hyperrectangles under the considered function and thus partitions the support of that function, while Lebesgue integration partitions the function output, so that it can be used to integrate functions with a support in more general spaces, such as manifolds. However, it requires the definition of a measure \cite[Sec. I.D]{choquet1982ama}, which is a generalized volume element, for example to measure the volume of the support where the function output belongs to a particular interval. A measure $\mu$, defined on a group $G$, is a Haar measure if 
\[\mu(g E) = \mu(E), \forall g \in G, \forall E \subset G,\]
where $E$ is a measurable subset of $G$ \cite[Sec. I.D]{choquet1982ama} and $gE = \{g.e: e \in E \}$. In this case, consider $E$ to be a subset of $\mathcal{M}_{n_T,n_T}$ and $G$ to be the unitary group of degree $n_T$. Consider the probability measure\footnote{A probability measure is a measure which integrates to one over the considered space.}
\[\mathrm{d}\mu = \frac{\mathrm{d}V}{\textrm{Vol}(\mathcal{M}_{n_T,n_T})},\]
where $\mathrm{d}V$ is the differential volume element of $\mathcal{M}_{n_T,n_T}$ ($\int_{\mathcal{M}_{n_T,n_T}} \mathrm{d}V = \textrm{Vol}(\mathcal{M}_{n_T,n_T})$) and $\textrm{Vol}(\mathcal{M}_{n_T,n_T})$ is the total volume of $\mathcal{M}_{n_T,n_T}$ \cite[Sec. 3.2]{edelman1989eac}. It can be shown that this measure is a Haar measure \cite[Sec. 2.1.4]{muirhead} \cite[Sec. 3.2]{edelman1989eac} \cite[Sec. 4.6]{edelman2005rmt}. Furthermore, this probability measure defines the pdf $p(V) = \frac{1}{\textrm{Vol}(\mathcal{M}_{n_T,n_T})}$, which is the uniform distribution in the Stiefel manifold. Hence, the considered Haar measure implies a uniform distribution, and there exists exactly one probability measure which is Haar. Therefore, matrices that are uniform within the Stiefel manifold with respect to the Haar measure are also denoted to be Haar-distributed \cite[Sec. 4.6]{edelman2005rmt}.

Note that Lebesgue integration is often avoided by using Jacobians, which actually determine the differential volume element $\mathrm{d}V$ \cite[Sec. 3.2]{edelman1989eac}. Similarly to the Stiefel manifold $\mathcal{M}_{n_T,n_T}$, one can also consider other Stiefel manifolds $\mathcal{M}_{i,j}$ for some $i$ and $j$.

The following lemma collects some useful properties of unitary matrices.

\begin{lemma}
	Consider the unitary matrices $D, U, V \in \C^{n_T \times n_T}$, where $U$ and $V$ are independent and uniformly distributed in $\mathcal{M}_{n_T,n_T}$, and $D$ is deterministic. Then
	\begin{enumerate}
		\item The linear transformation
	\[f: \C^n_T \rightarrow \C^n_T: f(\mb{b}) = U \mb{b} = \mb{a}, \]	
	is a bijection. Next, $\frac{\mb{a}}{||\mb{b}||}$ is uniformly distributed in $\mathcal{M}_{n_T,1}$ for any $\mb{b}$ independent of $U$.	
		\item $U^\prime=D U$ is uniform in $\mathcal{M}_{n_T,n_T}$.
		\item $U^\prime=V U$ is uniform in $\mathcal{M}_{n_T,n_T}$.
	\end{enumerate}
	\label{lemma: U is a bijection}
	\label{lemma: U properties}
\end{lemma}
\begin{proof}
	\begin{enumerate}
	\item Any non-singular matrix $U$ guarantees that $U \mb{b}_1 \neq U \mb{b}_2$ for $\mb{b}_1 \neq \mb{b}_2$, as $U^{-1} U \mb{b}_2 \neq \mb{b}_1$, yielding a bijection. \\
		Consider any $\mb{c}$ in $\mathcal{M}_{n_T,1}$. Consider $Q$ so that $\mb{c} = Q \frac{\mb{a}}{||\mb{b}||}$. Matrix $Q$ always exists; more specifically, $Q=C A^\dagger$, where $C$ and $A$ are unitary matrices with $\mb{c}$ and $\frac{\mb{a}}{||\mb{b}||}$ on its first column, respectively. Hence $\mb{c} = Q U \frac{\mb{b}}{|| \mb{b} ||}$, where $Q U$ is uniformly distributed in $\mathcal{M}_{n_T,n_T}$ (see third bullet). Hence $Q U$ can be considered as another realization of the random matrix $U$ (hence with the same probability density by the uniformity) and thus $\mb{c}$ as the corresponding realization of the random $\frac{\mb{a}}{||\mb{b}||}$, with the same probability density. Since this holds for any $\mb{c}$, all unitary vectors have the same probability density, yielding the uniform distribution in $\mathcal{M}_{n_T,1}$.	
	\item For each $D$, there exists exactly one $U$ yielding $U^\prime$ (because $D$ is a bijection), so that $p(U^\prime)=p(U)$.
	\item For each $U$, there exists exactly one $V$ yielding $U^\prime$ because $U$ is a bijection. Hence,
\begin{align}
	p(U^\prime) &= \int_{\mathcal{M}_{n_T,n_T}} p(V = U^\prime U^\dagger| U) p(U) d U \\
	&= \int_{\mathcal{M}_{n_T,n_T}} p(V) p(U) d U \\
	&= p(V) \int_{\mathcal{M}_{n_T,n_T}} p(U) d U = p(V),
\end{align}
	which yields the claim.
	\end{enumerate}
\end{proof}

\subsection{Proof of Lemma \ref{lemma: random prec toy ex}}
\label{sec: proof lemma toy example}

The diversity order is upper bounded by the diversity order of the outage probability, $P_{\out} = \int_0^{2 \pi} p(\theta) P_{\out | \theta} \mathrm{d} \theta$, where $p(\theta)=\frac{1}{2 \pi}$.
Consider $T_\theta = \gamma^{-0.5}$, then 
\begin{align}
	P_{\out} &\geq P_{\out | \theta = \theta_l} \int_0^{T_\theta} p(\theta)  \mathrm{d} \theta + \int_{T_\theta}^{2 \pi} p(\theta) P_{\out | \theta} \mathrm{d} \theta \\
	&\geq P_{\out | \theta = \theta_l} \int_0^{T_\theta}
        p(\theta) \mathrm{d} \theta ~ \doteq ~ P_{\out | \theta = \theta_l} \gamma^{-0.5} , \label{eq: out toy example}
\end{align}
where $\theta_l = \argmin{\theta \in [0,T_\theta]}~ P_{\out | \theta} $. Hence,
$P_{\out | \theta = \theta_l} \doteq \frac{1}{\gamma}$ is a sufficient condition to have that $P_{\out} \dot{\geq} \frac{1}{\gamma^{1.5}}$. 

Let us study the mutual information when $\theta \in [0,T_\theta]$. Similarly as in Eq. \ref{eq: mut info discrete alphabet MIMO form 1}, we express the mutual information between BPSK input and output of the parallel channel in Eq. (\ref{eq: mut info discrete alphabet toy example form 1}), where $\circ$ denotes the component-wise multiplication, also known as the Hadamard product or Schur product, and $\bs{\beta}=[\beta_1 ~\beta_2]$. As was first done in \cite{tse2003dam} and then later in \cite[App. I]{fab2007cmi} for Rayleigh fading, we define the normalized fading gains $\alpha_i = -\frac{\log \beta_i^2}{\log \gamma}$, so that the mutual information can be expressed as in Eq. (\ref{eq: mut info discrete alphabet toy example}), where
\begin{equation}
f(\alpha_i, s_{t,i}, w_{t,i}) = e^{-\gamma^{1-\alpha_i} s_{t,i}^2 - 2 \sqrt{\gamma^{1-\alpha_i}} w_{t,i} s_{t,i}},
\end{equation}
and where $s_{t,i} = (x_{t,i} - x^{\prime}_{t,i})$. 
\begin{figure*}
\begin{align}
	I\left(\mb{x}_t; \mb{y}_t| \beta_1, \beta_2, \theta \right) &=  2 - \frac{1}{4} \sum_{\mb{x}_t \in \Omega_\mb{x}}  \E_{\mb{y}_t|\mb{x}_t} \left[  \Log_2 \left( \sum_{\mb{x}_t^{\prime} \in \Omega_\mb{x}} \exp\left[ \left( d^2(\mb{y}_t,\sqrt{\gamma} \bs{\beta} \circ \mb{x}_t) - d^2(\mb{y}_t, \sqrt{\gamma} \bs{\beta} \circ  \mb{x^{\prime}}_t) \right) \right] \right) \right] \label{eq: mut info discrete alphabet toy example form 1} \\
	I\left(\mb{x}_t; \mb{y}_t| \beta_1, \beta_2, \theta \right) &=  2 - \frac{1}{4} \sum_{\mb{x}_t \in \Omega_\mb{x}}  \E_{\mb{w}_t} \left[ \Log_2 \left( \sum_{\mb{x}_t^\prime \in \Omega_\mb{x}} \prod_{i=1}^{2} f(\alpha_i, s_{t,i}, w_{t,i}) \right) \right], \label{eq: mut info discrete alphabet toy example}
\end{align}
\end{figure*}

We recall the distribution of $\bs{\alpha}=(\alpha_1, \alpha_2)$ \cite[App. I]{fab2007cmi},
\begin{equation}
p(\bs{\alpha}) = (\log \gamma)^2 e^{-\gamma^{-\alpha_1}-\gamma^{-\alpha_2}} \gamma^{-\alpha_1 - \alpha_2}.
\end{equation}
We denote $\alpha_i$ as an $\epsilon$-bad fading gain when $\alpha_i \geq 1+\epsilon$ and as an $\epsilon$-good fading gain when $\alpha_i \leq 1-\epsilon$, for $\epsilon > 0$.

The outline of the remainder of the proof is as follows. We first determine a region $\mathcal{A}$ in the space of $(\alpha_1, \alpha_2)$ where the mutual information, given a fading point in that region, is smaller than the spectral efficiency (the spectral efficiency $R = m R_c$ equals $2 R_c$ in this case) corresponding with a coding rate $R_c > 0.5$. The outage probability will then be given by the integral $\int_\mathcal{A} p(\bs{\alpha}) \mathrm{d} \bs{\alpha}$. As we are only interested in the SNR-exponent of the outage probability, we can simplify the integral following the same lines as in \cite[Sec. III-B]{tse2003dam}, \cite{fab2007cmi}, by replacing $p(\bs{\alpha})$ by $q(\bs{\alpha})$, 
\begin{equation}
q: \R^{2,+} \rightarrow \R: q(\bs{\alpha}) = \gamma^{-\alpha_1 - \alpha_2},
\end{equation}
because $\log(\gamma)^2 \dot{=} 1$, $e^{-\gamma^{-\alpha_1}}$ decreases exponentially with $\gamma$ if $\alpha_i < 0$ and approaches $1$ and $e$ for $\alpha_i > 0$ and $\alpha_i=0$ respectively, so that we only have to consider $\bs{\alpha} \in \R^{2,+}$ where we can drop the exponential term. The set $\R^{2,+}$ corresponds to all positive real-valued two-dimensional vectors. Hence, 
\begin{equation}
\int_\mathcal{A} p(\bs{\alpha}) \mathrm{d} \bs{\alpha} \dot{=} \int_{\mathcal{A} \cap \R^{2,+}} q(\bs{\alpha}) \mathrm{d} \bs{\alpha}.
\end{equation}

As in \cite{fab2007cmi}, we can apply the dominated convergence theorem \cite[Theorem 1.6.7]{dur1991pta} to determine the diversity order (Eq. (\ref{eq: div order}))
\begin{multline}
	\lim_{\gamma \rightarrow \infty} \E_{\mb{w}_t} \left[ \Log_2 \left( \sum_{\mb{x}^{\prime}_t \in \Omega_\mb{x}} \prod_{i=1}^2 f(\alpha_i, s_{t,i}, w_{t,i}) \right) \right] \\
	= \E_{\mb{w}_t} \left[ \lim_{\gamma \rightarrow \infty} \Log_2 \left( \sum_{\mb{x}^{\prime}_t \in \Omega_\mb{x}} \prod_{i=1}^2 f(\alpha_i, s_{t,i}, w_{t,i}) \right) \right]
\end{multline}

When $\theta \rightarrow 0$, we have that for each $\mb{x}_t$ and $i$,
\begin{equation}
	\min_{\mb{x}_t^\prime} |s_{t,i}| = \sqrt{2} (\theta - O(\theta^3)).
\end{equation}
More specifically, for each $\mb{x}_t$ and $i$, there exists exactly one $\mb{x}_t^\prime \neq \mb{x}_t$ so that $|s_{t,i}| < \sqrt{2} T_\theta$ when $\theta \in [0, T_\theta]$ \footnote{Note that it is impossible to have $|s_{t,i}| \rightarrow 0,$ for $i=1,2$ (see Lemma \ref{lemma: no $n_T$ zero elements}).}. As a consequence, $|s_{t,i}| \dot{\leq} T_\theta$ and it is easy to verify that $\lim_{\gamma \rightarrow \infty} f(\alpha_i, |s_{t,i}| \dot{\leq} T_\theta, w_{t,i}) = 1$ if $\alpha_i > 0$. The mutual information for large $\gamma$ can now be written as
\[I\left(\mb{x}_t; \mb{y}_t | \beta_1, \beta_2, \theta \right) =  2 - \frac{1}{4} \sum_{\mb{x}_t \in \Omega_\mb{x}}  \E_{\mb{w}_t} \left[ \Log_2 \left( 1+ g(\bs{\alpha}, \mb{x}_t) \right) \right],\]
where, for all $\mb{x}_t$,
\begin{equation}
g(\bs{\alpha}, \mb{x}_t) = \left\{ 
	\begin{array}{lr} 
		O\left( \exp\left[-\gamma^{\epsilon}\right] \right), & \textrm{if } \alpha_1, \alpha_2 < 1, \\
		\Omega\left( \exp\left[-\gamma^{-\epsilon}\right] \right) & \textrm{if } \alpha_1 > 1 \textrm{ or } \alpha_2 > 1,
	\end{array}
	\right. 
\end{equation}
where $\epsilon = |1-\max_i \alpha_i|$. Hence, if there exists one $\epsilon$-bad fading gain, $\alpha_i \geq 1+\epsilon$, then it can be shown that the mutual information is $1-\Omega \left(\gamma^{-\epsilon}  \right)$, thus strictly smaller than one. Therefore, 
$\Prob(I\left(\mb{x}_t; \mb{y}_t | \bs{\alpha}, \theta \leq T_\theta \right) < 2 R_c) = 1$
when $R_c > 0.5$ and $\bs{\alpha} \in \mathcal{A}_\epsilon = \{\bs{\alpha}: \sum_{i=1}^2 \ind\{\alpha_i \geq 1+\epsilon\} \geq 1 \}$ ($\ind\{.\}$ is the indicator function). As a consequence, for every $\epsilon>0$ and $R_c > 0.5$, we have that
\begin{equation}
	P_{\out | \theta = \theta_l} \geq \Prob(\mathcal{A}_\epsilon) \dot{=} \int_{\mathcal{A}_\epsilon \cap \R^{2,+}} q(\bs{\alpha}) \mathrm{d} \bs{\alpha}.
\end{equation}
Following the same lines as in \cite[Theorem 4]{tse2003dam}, \cite[App. I]{fab2007cmi}, the SNR exponent $d_{\out|\theta = \theta_l}$ is 
\begin{equation}
	d_{\out|\theta = \theta_l} \leq d_{\out}(\epsilon) = \inf_{\mathcal{A}_\epsilon \cap \R^{2,+}} \sum_{i=1}^2 \alpha_i = 1+\epsilon.
	\label{eq: div order toy example}
\end{equation}
This holds for any $\epsilon>0$, and the bound in Eq. (\ref{eq: div order toy example}) can be made tight taking the infinum $\inf_\epsilon d_\out(\epsilon)$ (see e.g. \cite{fab2007cmi, nguyen2007atl, fab2007mcm}). Combining with (\ref{eq: out toy example}), we obtain $d_\out \leq 0.5+ d_{\out|\theta = \theta_l} \leq 1.5$.

\subsection{Bad precoders for $\Omega_z=4$-QAM and $2 \times 2$ MIMO}
\label{app: bad precoders for QPSK}
 
Given $\Omega_z$, all bad precoders can be identified, for example by listing all possible difference vectors $\mb{z} - \mb{z}^\prime$ and consequently determining the structure of the precoder so that $\exists~ i$ where $s_{t,i}=0$. It can be verified that for $\Omega_z = 4-$QAM (corresponding to $m=4$), the \textit{bad} precoders $V_\bad$ is the set in Eq. (\ref{eq: bad precoders 4-QAM}), where $\omega, \psi, \phi, \theta \in [0, 2 \pi]$ and where 
\begin{eqnarray}
	\delta &=& \psi + k_\phi \frac{\pi}{2}, ~~~ k_\phi \in \{0, \ldots, 3\} \\
	\eta &=& \psi + \frac{\pi}{4} + k_\phi \frac{\pi}{2}, ~~~ k_\phi \in \{0, \ldots, 3\} \\
	\iota &=& \theta + (k_\phi-2) \frac{\pi}{2}.
\end{eqnarray}
\begin{figure*}
	\begin{equation}
	V_\bad  = \left\{ \begin{bmatrix} 0 & e^{j\omega} \\ e^{j\psi} & 0 \end{bmatrix}, \begin{bmatrix} e^{j\phi} & 0 \\ 0 & e^{j\theta} \end{bmatrix}, \frac{1}{\sqrt{2}} \begin{bmatrix} e^{j\delta} & e^{j\iota} \\ e^{j\psi} &e^{j\theta} \end{bmatrix}, \frac{1}{\sqrt{3}} \begin{bmatrix} e^{j\eta} & \sqrt{2} e^{j\iota} \\ \sqrt{2} e^{j\psi} &e^{j\theta} \end{bmatrix}, \frac{1}{\sqrt{3}} \begin{bmatrix} \sqrt{2} e^{j\eta} & e^{j\iota} \\ e^{j\psi} & \sqrt{2} e^{j\theta} \end{bmatrix} \right\}.
	\label{eq: bad precoders 4-QAM}
	\end{equation}
\end{figure*}
For each precoder $V_t \in V_\bad$, the maximum mutual information $I(\mb{x}_t;\mb{y}_t|V_t)$ can be determined. It can be verified that for the first two types of \textit{bad} precoders in Eq. (\ref{eq: bad precoders 4-QAM}), the maximal mutual information is $\frac{m}{2} = 2$. This laborious study of bad precoders is only necessary to analyze the case with the time-varying precoder when a small finite number of unitary matrices is used. Fortunately, the average mutual information when $N$ precoders are used rapidly converges to the expectation of the mutual information, $\E_t[I(\mb{x}_t;\mb{y}_t|V_t)]$, which corresponds to $N \rightarrow \infty$ (see Sec. \ref{sec: Maximal coding rate for EMI-N code}). In Theorem \ref{prop: full diversity mimo with distr rotation}, we assume that $N \rightarrow \infty$ and thus only consider $\E_t[I(\mb{x}_t;\mb{y}_t|V_t)]$. The convergence properties to this case for finite $N$ are discussed in Sec. \ref{sec: Maximal coding rate for EMI-N code}.

\subsection{Probability of corrupt precoders}
\label{app: probability of corrupt precoders}

Before we determine the probability of corrupt precoders, we consider the following Lemma. 
\begin{lemma}
	\label{lemma: beta distribution}
	Consider two independent random variables $X$ and $Y$ which are up to a constant $\chi^2$-distributed with parameter $2 a$ and $2 b$, respectively. Then the random variable 
\[Z = \frac{X}{X+Y} \sim \beta(a,b), \]
	where $\beta(a,b)$ is the Beta-distribution with parameters $a$ and $b$, respectively.
\end{lemma}
\begin{proof}
	$X \sim \Gamma(a, \theta=2) \Leftrightarrow X \sim \chi^2(2a)$. 
	A property of Beta-distributions \cite[Chapt. 25]{johnson1995cud} is that if $X \sim \Gamma(a, \theta=2)$ and $Y \sim \Gamma(b, \theta=2)$, then $Z = \frac{X}{X+Y} \sim \beta(a,b)$.
\end{proof}

Consider the unitary matrix $V_t$, which is uniformly distributed in the set of all unitary matrices $\mathcal{M}(n_T,n_T)$  (see App. \ref{app: Haar}). Consider any pair $\mb{z}, ~\mb{z}^\prime \in \Omega_{\mb{z}}$ and define $\mb{v}_t = V_t \frac{\mb{z}-\mb{z}^\prime}{|| \mb{z}-\mb{z}^\prime ||}$, so that $\mb{s}_t = || \mb{z}-\mb{z}^\prime || \mb{v}_t$. It follows from Lemma \ref{lemma: U properties} that $\mb{v}_t$ is uniform in $\mathcal{M}_{n_T,1}$. As a consequence of the uniform distribution, it is shown in \cite[Theorem 3.1]{edelman1989eac} that $\mb{v}_t$ can be constructed as 
\begin{equation}
\label{eq: construction of random v}
	\mb{v}_t = \frac{\mb{g}}{||\mb{g}||},
\end{equation} 
where $\mb{g} \sim \Co\N(0,I)$. 
Hence, 
\[\sum_{i=2}^{n_T} |v_{t,i}|^2 = \frac{\sum_{i=2}^{n_T} |g_i|^2}{ \sum_{i=1}^{n_T} |g_i|^2} ~~\textrm{ and }~~ |v_{t,i}|^2 = \frac{|g_i|^2}{ \sum_{i=1}^{n_T} |g_i|^2},\] 
where $\mb{g} \sim \Co\N(0,I)$. Because $|g_i|^2 \sim \chi^2(2)$ and $\sum_{i=2}^{n_T} |g_i|^2 \sim \chi^2(2(n_T-1))$ (up to a constant), we have that $\sum_{i=2}^{n_T} |v_{t,i}|^2 \sim \beta(n_T-1,1)$ and $|v_{t,i}|^2 \sim \beta(1, n_T-1)$ by Lemma \ref{lemma: beta distribution}. Now, we can determine the probability that $\{|s_{t,i}|^2 \leq \gamma^{-1},~ i>1\}$ for this given pair $\mb{z}, ~\mb{z}^\prime \in \Omega_{\mb{z}}$, denoted by $\textrm{PCP}$ (pairwise corruption probability), as
\begin{equation}
	\textrm{PCP} = \Prob(|s_{t,2}|^2 \leq \gamma^{-1}, \ldots, |s_{t,n_T}|^2 \leq \gamma^{-1})
\end{equation}
The PCP is upper and lower bounded as follows
\begin{align}
	\Prob \left(\sum_{i=2}^{n_T} |s_{t,i}|^2 \leq \gamma^{-1} \right) &\leq \textrm{PCP} \nonumber \\ 
	&\leq \Prob \left(\sum_{i=2}^{n_T} |s_{t,i}|^2 \leq (n_T-1) \gamma^{-1} \right).
\end{align}
Because $\mb{v}_t$ is proportional to $\mb{s}_t$, we have
\begin{align}
	\Prob \left(\sum_{i=2}^{n_T} |v_{t,i}|^2 \leq \gamma^{-1} \right) &\dot{\leq} \textrm{PCP} \nonumber \\
	&\dot{\leq} \Prob \left(\sum_{i=2}^{n_T} |v_{t,i}|^2 \leq (n_T-1) \gamma^{-1} \right),
\end{align}
which is 
\begin{equation}
	\gamma^{-(n_T-1)} \dot{\leq} \textrm{PCP} \dot{\leq} \gamma^{-(n_T-1)},
\end{equation}
by App. \ref{app: Asymptotic cumulative distribution function of Beta distribution}. Hence, 
\begin{equation}
	\textrm{PCP} \doteq \gamma^{-(n_T-1)}.
\end{equation}

The probability of a corrupt precoder is bounded as 
\begin{equation}
	\textrm{PCP} \leq \Prob(\Set_{c,1}) \leq \frac{2^m (2^m-1)}{2} \textrm{PCP} \label{bound prob of corrupt precoder}
\end{equation}
yielding
\begin{equation}
	\Prob(\Set_{c,1})  \doteq \gamma^{-(n_T-1)}.	
\end{equation}
The lower bound in (\ref{bound prob of corrupt precoder}) is because the PCP only considers one of the $\frac{2^m (2^m-1)}{2}$ pairs $\mb{z}, ~\mb{z}^\prime \in \Omega_{\mb{z}}$. The upper bound neglects the correlation between the different pairs $\mb{z}, ~\mb{z}^\prime \in \Omega_{\mb{z}}$.

The proof for $\Prob(\Set_{c,2})$ follows the same lines by considering the components $\{s_{t,i}, i=1, \ldots, r\}$ instead of $\{s_{t,i}, i=2, \ldots, n\}$.

\subsection{Asymptotic cdf of Beta distribution}
\label{app: Asymptotic cumulative distribution function of Beta distribution}

We are interested in $\Prob(y<x)$ for small $x$, where $y \sim \beta(1, n_T-1)$ or $y \sim \beta(n_T-1, 1)$. The cdf of a Beta-distributed random variable $Y$ with parameters $a$ and $b$ is
\begin{equation}
	F_Y(x, a, b) = \int_0^x \frac{\Gamma(a+b)}{\Gamma(a)\Gamma(b)} x^{a-1} (1-x)^{b-1} \mathrm{d}x.
\end{equation}
Hence, 
\begin{align}
	F_Y(x, 1, n_T-1) &= -(1-x)^{n_T-1} +1 = (n_T-1) x - O(x^2) \label{cdf Beta 1}\\
	F_Y(x, n_T-1, 1) &= x^{n_T-1}
\end{align}
where it is assumed that $|x|<1$ in (\ref{cdf Beta 1}).

\subsection{Proof of Lemma \ref{lemma: no full diversity mimo}}
\label{app: proof lemma no full div mimo}

The achievability is shown by simply using one transmit antenna and maximum-ratio combining (ML detection) at the receiver. For the converse, we lower bound the outage probability as follows.
\begin{align}
	P_{\out} &= \int_{\Set_c} p(V^\prime) P_{\out | V^\prime}  \mathrm{d} V^\prime + \int_{\bar{\Set_c}} p(V^\prime)   P_{\out | V^\prime} \mathrm{d} V^\prime \\
	&\geq P_{\out | V^\prime = V_l} \int_{\Set_c} p(V^\prime)  \mathrm{d} V^\prime  = P_{\out | V^\prime = V_l} \Prob(\Set_c)
\end{align}
where $\bar{\Set_c}$ is the complement of $\Set_c$ and $V_l = \argmin{V^\prime \in \Set_c}~ P_{\out | V^\prime}$.

First, let us study $P_{\out | V^\prime = V_l}$, which is the outage probability given that $V^\prime = V_l$. The outline of the remainder of the proof is the same as in the proof of Lemma \ref{lemma: random prec toy ex} (App. \ref{sec: proof lemma toy example}). We determine a region $\mathcal{A}$ in the space of $\bs{\alpha}$ where the mutual information (given in Eq. \ref{eq: mut info discrete alphabet MIMO}) is smaller than a particular spectral efficiency $R$. The outage probability, given by $\int_\mathcal{A} p(\bs{\alpha}) \mathrm{d} \bs{\alpha}$, can then be simplified following the same lines as in \cite[Sec. III-B]{tse2003dam}, \cite{fab2007cmi} because we are only interested in its SNR-exponent. More specifically, it can be verified that \cite[Sec. III-B]{tse2003dam}, \cite{fab2007cmi}
\begin{equation}
\int_\mathcal{A} p(\bs{\alpha}) \mathrm{d} \bs{\alpha} \dot{=} \int_{\mathcal{A} \cap \R^{\min(n_T,{n_R}),+}} q(\bs{\alpha}) \mathrm{d} \bs{\alpha},
\end{equation}
where 
\begin{align}
& q: \R^{\min(n_T,{n_R}),+} \rightarrow \R: \nonumber \\
& q(\bs{\alpha}) = \prod_{i=1}^{\min({n_R},n_T)} \gamma^{-(2 i -1 +|{n_R}-n_T|)\alpha_i}.
\end{align}

In order to determine the diversity order (Eq. (\ref{eq: div order})), we can apply the dominated convergence theorem \cite[Theorem 1.6.7]{dur1991pta}, \cite[App. 2]{fab2007cmi}, so that 
\begin{multline}
	\lim_{\gamma \rightarrow \infty} \E_{\mb{w}} \left[ \Log_2 \left( \sum_{\mb{x}^\prime \in \Omega_\mb{x}} \prod_{i=1}^{n_T} f(\alpha_i, s_{i}, w_{i}) \right) \right] \\
	= \E_{\mb{w}} \left[ \lim_{\gamma \rightarrow \infty} \Log_2 \left( \sum_{\mb{x}^\prime \in \Omega_\mb{x}} \prod_{i=1}^{n_T} f(\alpha_i, s_{i}, w_{i}) \right) \right]
\end{multline}

Consider first the case that ${n_R} \geq n_T$. For corrupt precoders, $\exists~\mb{z}_c, \mb{z}_c^\prime \neq \mb{z}_c$, satisfying $|s_{i}| \leq \gamma^{-0.5}, \forall i > 1$. In the case of a corrupt precoder, we observe that for large $\gamma$ and $i >1$, $\lim_{\gamma \rightarrow \infty} f(\alpha_i, s_{i}, w_{i}) = 1$ if $\alpha_i > 0$. For large $\gamma$, the instantaneous mutual information is 
\[I\left(\mb{x}; \mb{y} | \Sigma, V^\prime \in \Set_{c,1} \right) =  m - 2^{-m} \sum_{\mb{x} \in \Omega_\mb{x}}  \E_{\mb{w}} \left[ \Log_2 \left( 1+ g(\bs{\alpha}, \mb{x}) \right) \right],\]
where 
\[g(\bs{\alpha}, \mb{x}_c) = \left\{ 
	\begin{array}{lr} 
		O\left( e^{-\gamma^{1-\alpha_1}} \right), & \textrm{if } \alpha_{1} < 1 \\
		\Omega\left( e^{-\left(\frac{1}{\gamma} \right)^{\alpha_{1}-1}} \right), & \textrm{if } \alpha_1 > 1
	\end{array}
	\right. ,
\]
and $\mb{x}_c = V^\prime \mb{z}_c$. Note that $g(\bs{\alpha}, \mb{x}) = \sum_{\mb{x}^\prime \neq \mb{x}} \prod_{i=1}^{n_T} f(\alpha_i, s_{i}, w_{i})$.

As a consequence, in the event that $\bs{\alpha} \in \mathcal{A}_{\epsilon, 1}$, the instantaneous mutual information is $m-\Omega\left(\exp\left[-\gamma^{-\epsilon}\right] \right)$, where 
\begin{equation}
\mathcal{A}_{\epsilon, 1} = \{ \bs{\alpha}: \alpha_1 \geq \ldots \geq \alpha_{n_T} > 0: \alpha_1 \geq 1 + \epsilon \}, 
\end{equation}
and $\epsilon > 0$. Hence, there exists a coding rate so that the spectral efficiency $m R_c$ is always larger than or equal to the mutual information, given that $\bs{\alpha} \in \mathcal{A}_{\epsilon, 1}$. For this coding rate, the outage probability is
\begin{equation}
	P_{\out | V^\prime \in \Set_{c,1}} \geq \int_{\mathcal{A}_{\epsilon, 1}} p(\bs{\alpha}) \mathrm{d} \bs{\alpha} \dot{=} \int_{\mathcal{A}_{\epsilon, 1}} q(\bs{\alpha}) \mathrm{d} \bs{\alpha}.
	\label{eq: b_epsilon outage}
\end{equation}
Noting that $\Prob(\Set_{c,1}) \dot{=} \gamma^{-(n_T-1)}$ and thus $P_{\out} \dot{\geq} \frac{1}{\gamma^{n_T-1}} P_{\out | V^\prime \in \Set_{c,1}}$; and following the same lines as in \cite[Theorem 4]{tse2003dam}, \cite{fab2007mcm}
, we obtain the SNR-exponent
\begin{equation}
	d_\out \leq d_\out(\epsilon) =  n_T-1+ \inf_{\bs{\alpha} \in \mathcal{A}_{\epsilon,1}} \sum_{i=1}^{n_T} (2 i -1 + {n_R}-n_T) \alpha_i.
	\label{eq: div order lemma}
\end{equation}
The infinum is $({n_R}-n_T+1)(1+\epsilon)$, which is achieved when $\alpha_{n_T} = \ldots = \alpha_2 = 0$ and $\alpha_1 = 1+\epsilon$. This holds for each $\epsilon>0$, and the bound in Eq. (\ref{eq: div order lemma}) can be made tight taking the infinum $\inf_\epsilon d_\out(\epsilon)$ (see 
\cite{fab2007mcm}), hence we obtain $d_\out \leq {n_R}$.	

Now consider the case that ${n_R} < n_T$. For corrupt precoders, $\exists~\mb{z}_c, \mb{z}^\prime \neq \mb{z}_c$, satisfying $|s_i|^2 \leq \gamma^{-1-\epsilon_i},~ \epsilon_i \geq 0, ~\forall~ i =1, \ldots, {n_R}$. In the case of a corrupt precoder, we observe that for large $\gamma$, $\lim_{\gamma \rightarrow \infty} f(\alpha_i, s_i, w_i) = 1$ if $\alpha_i > 0$. For large $\gamma$, the instantaneous mutual information is 
\[I\left(\mb{x}; \mb{y} | \Sigma, V^\prime \in \Set_{c,2} \right) =  m - 2^{-m} \sum_{\mb{x} \in \Omega_\mb{x}}  \E_{\mb{w}} \left[ \Log_2 \left( 1+ g(\bs{\alpha}, \mb{x}) \right) \right],\]
where 
\begin{equation}
g(\bs{\alpha}, \mb{x}_c) = \Omega\left( e^{-\left(\frac{1}{\gamma} \right)^{\max_i (\alpha_{i}+\epsilon_i)}} \right), ~~~ \mb{x}_c = V^\prime \mb{z}_c.
\end{equation}

Hence, for large $\gamma$, the mutual information is therefore $m-\Omega\left(e^{-\gamma^{-\max_i (\alpha_{i}+\epsilon_i)}} \right)$, for all $\bs{\alpha} \in \R^{{n_R},+}$. Hence, there exists a coding rate so that the spectral efficiency $m R_c$ is always larger than or equal to the mutual information. For this coding rate, the outage probability $P_{\textrm{out} | V^\prime \in \Set_{c,2}} \dot{=} 1$. Noting that $\Prob(\Set_{c,2}) \dot{=} \gamma^{-{n_R}}$ and thus $P_{\textrm{out}} \dot{\geq} \frac{1}{\gamma^{{n_R}}} P_{\textrm{out} | V^\prime \in \Set_{c,2}}$, we obtain that $d_\out \leq {n_R}$.

\subsection{Proof of Theorem \ref{prop: full diversity mimo with distr rotation}}
\label{app: full diversity mimo with distr rotation}

Similarly to $\Set_c$, we define a larger set $\Set_{c,3}$, which is the set of precoders $V_t$ so that $\exists~i \leq \min(n_T,{n_R}), \mb{z}, \mb{z}^\prime \neq \mb{z}$, satisfying $|s_{t,i}|^2 \leq (\log \gamma)^{-p}$, for any $p>0$. The probability $\Prob(V_t \in \Set_{c,3}) \rightarrow 0$ for large $\gamma$ (see App. \ref{app: probability second corrupt precoder set}). Denoting $I(\mb{x}_t; \mb{y}_t | \Sigma, V_t)$ as $I(V_t)$, we can write $\E_t \left[ I(V_t) \right]$ as
\begin{align}
\E_t \left[ I(V_t) \right] &=	\int_{\Set_{c,3}} p(V_t) I(V_t) \mathrm{d}V_t + \int_{\bar{\Set}_{c,3}} p(V_t) I(V_t) \mathrm{d}V_t. \label{sample mean in app}\\
 &\geq I_{\inner} = I(V_l) \int_{\bar{\Set}_{c,3}} p(V_t) \mathrm{d}V_t =I(V_l) (1-\Prob(\Set_{c,3})), \nonumber
\end{align}
where  $V_l = \argmin{V_t \in \bar{\Set}_{c,3}} ~ I(V_t)$ (worst case). Thus, we have that 
\begin{equation}
	P_\out = \Prob(\E_t \left[ I(V_t) \right] \leq R) \leq \Prob(I_\inner \leq R),
\end{equation}
By definition of $\bar{\Set}_{c,3}$, $|s_{t,i}| > (\log \gamma)^{-p}$, so that for large $\gamma$,
\begin{equation}
	f(\alpha_i, s_{t,i}, w_{t,i}) = 
	\left\{ 
	\begin{array}{lr} 
	O\left( e^{-\gamma^{1-\alpha_i}} \right), & \alpha_i < 1 \\
	\Omega\left( e^{-\left( \frac{1}{\gamma} \right)^{\alpha_i-1}} \right), & \alpha_i>1.
	\end{array}
	\right. 
	\label{eq: good or bad fading gains}
\end{equation}
Hence, if $\sum_{i=1}^{\min(n_T,{n_R})} \ind\{\alpha_i <1\} \geq 1$ ($\ind\{.\}$ is the indicator function), then $I(V_l) \rightarrow m$ for large $\gamma$. More specifically, $I(V_l) \rightarrow m$ for large $\gamma$ when $\bs{\alpha} \in \mathcal{A}_\epsilon$, where
\begin{equation}
\mathcal{A}_\epsilon = \{\bs{\alpha}: \sum_{i=1}^{\min(n_T,{n_R})} \ind\{\alpha_i <1-\epsilon\} \geq 1 \},
\end{equation}
for $\epsilon>0$. Note that the fading gains are ordered, so that this set is equivalent to 
\[\mathcal{A}_\epsilon = \{\bs{\alpha}: \alpha_1 \geq \ldots \geq \alpha_{\min(n_T,{n_R})} > 0: \alpha_{\min(n_T,{n_R})} < 1-\epsilon \}, \epsilon>0.\]
As a consequence, $I(V_l | \bs{\alpha} \in \mathcal{A}_\epsilon) = m - O \left(e^{-\gamma^\epsilon} \right)$. Note that for large $\gamma$ and $R_c = 1-\epsilon_2$, $\epsilon_2 > 0$, 
\begin{equation}
	\lim_{\gamma \rightarrow \infty} \Prob(I(V_l | \bs{\alpha} \in \mathcal{A}_\epsilon) (1-\Prob(\Set_{c,3}) \leq R_c m) = 0.
	\label{eq: maximal mutual info full div emi}
\end{equation}
In other words, $\Prob(I_\inner \leq R) \leq \Prob(\bar{\mathcal{A}_\epsilon})$. Hence, for large $\gamma$ and any $\epsilon, \epsilon_2>0$, the outage probability is upper bounded by $\Prob(\bar{\mathcal{A}_\epsilon})$,
where $\bar{\mathcal{A}_\epsilon} = \{ \bs{\alpha}: \sum_{i=1}^{\min(n_T,{n_R})} \ind\{\alpha_i \geq 1-\epsilon\} = {\min(n_T,{n_R})} \}$ or 
$\bar{\mathcal{A}_\epsilon} = \{ \bs{\alpha}: \alpha_1 \geq \ldots \geq \alpha_{\min(n_T,{n_R})} > 0: \alpha_{\min(n_T,{n_R})} \geq 1-\epsilon \}.$ 
Following the same lines as before (or see \cite{fab2007mcm, tse2003dam}), 
\begin{equation}
	d_\out \geq \sup_{\epsilon} \inf_{\bs{\alpha} \in \bar{\mathcal{A}_\epsilon}} \sum_{i=1}^{\min(n_T,{n_R})} (2i -1 + |{n_R}-n_T|) \alpha_i 
\end{equation}
By letting $\epsilon_2 \rightarrow 0^+$, we obtain $d_\out = n_T {n_R}$ for any $R_c < 1$.

\subsection{Probability of $\Set_{c,3}$}
\label{app: probability second corrupt precoder set}

Similarly as in App. \ref{app: probability of corrupt precoders}, we consider $\mb{v}_t = V_t \frac{\mb{z}-\mb{z}^\prime}{|| \mb{z}-\mb{z}^\prime ||}$, uniform in the Stiefel manifold $\mathcal{M}_{n_T,1}$, so that $\mb{s}_t = || \mb{z}-\mb{z}^\prime || \mb{v}_t$; and we consider the PCP so that 
\begin{equation}
	\Prob(\Set_{c,3}) \leq \frac{2^m (2^m-1)}{2} \textrm{PCP} \label{bound prob of corrupt precoder Sc3}
\end{equation}
where in this case
\begin{equation}
	\textrm{PCP} = \Prob \left(|s_{t,i}|^2 \leq (\log \gamma)^{-p} \right)
\end{equation}
for at least one $i \in [1, \ldots, \min(n_T, n_R)]$. Hence
\begin{align}
	\textrm{PCP} &\leq \sum_{i=1}^{\min(n_T,{n_R})} \Prob \left( |s_{t,i}|^2 \leq (\log \gamma)^{-p} \right) \\
	&{\dot{=}}  \sum_{i=1}^{\min(n_T,{n_R})} \Prob \left( |v_{t,i}|^2 \leq (\log \gamma)^{-p} \right) \label{eq 2 in app S_c3} \\
	&= \sum_{i=1}^{\min(n_T,{n_R})} (n_T-1) (\log \gamma)^{-p} - O((\log \gamma)^{-2p} ), \label{eq: Sc prob}
\end{align}
where (\ref{eq 2 in app S_c3}) follows from $\mb{s}_t$ being proportional to $\mb{v}_t$, where (\ref{eq: Sc prob}) follows from App. \ref{app: Asymptotic cumulative distribution function of Beta distribution} noting that $|v_{t,i}|^2 \sim \beta(1, n_T-1)$ (see App. \ref{app: probability of corrupt precoders}). Hence, for large $\gamma$, $\Prob(\Set_{c,3}) \rightarrow 0$.



\end{document}